\documentclass[a4paper,reqno]{amsart}

\usepackage[margin=2cm,papersize={19.75cm,25cm}]{geometry}
\usepackage[all]{xy}           
\usepackage{amssymb}           
\usepackage{hyperref}
\usepackage{eucal}
\usepackage[dvips]{graphics}
\usepackage{graphicx}
\usepackage{exscale}
\usepackage{color}
\usepackage{tikz}
\usetikzlibrary{matrix}
\usepackage{subfig}
\DeclareGraphicsExtensions{.jpg,.png,.pdf}
\usepackage{exscale}
\usepackage{pgfplots}
\usepackage{tikz}
\usetikzlibrary{arrows,chains,matrix,positioning,scopes}
\makeatletter
\tikzset{join/.code=\tikzset{after node path={%
\ifx\tikzchainprevious\pgfutil@empty\else(\tikzchainprevious)%
edge[every join]#1(\tikzchaincurrent)\fi}}}
\makeatother
\tikzset{>=stealth',every on chain/.append style={join},
         every join/.style={->}}
\tikzstyle{labeled}=[execute at begin node=$\scriptstyle,
   execute at end node=$]

\numberwithin{equation}{section}

\newtheorem{definition}{Definition}[section]

\newtheorem{proposition}[definition]{Proposition}
\newtheorem{corollary}[definition]{Corollary}
\newtheorem{remarkth}[definition]{Remark}

\newenvironment{remark}{\begin{remarkth}\upshape}{\hfill$\diamond$\end{remarkth}}

\renewcommand{\emph}[1]{{\bfseries\itshape{#1}}}

\newcommand{\lcf}{\lbrack\! \lbrack}
\newcommand{\rcf}{\rbrack\! \rbrack}

\newcommand{\lvec}[1]{\overleftarrow{#1}}
\newcommand{\rvec}[1]{\overrightarrow{#1}}

\makeatletter
\newcommand\prol{\@ifstar{\@proldf}{\@prolpf}}  
\def\@prolpf{\@ifnextchar[{\@prolpf@wrt}{\@prolpf@}}
\def\@prolpf@wrt[#1]#2{\@ifnextchar[{\@prolpf@wrt@at{#1}{#2}}{\@prolpf@wrt@{#1}{#2}}}
\def\@prolpf@wrt@at#1#2[#3]{\prolsymbol^{#1}_{#3}#2}
\def\@prolpf@wrt@#1#2{\prolsymbol^{#1}#2}
\def\@prolpf@#1{\@ifnextchar[{\@prolpf@at{#1}}{\@prolpf@@{#1}}}
\def\@prolpf@at#1[#2]{\prolsymbol_{#2}#1}
\def\@prolpf@@#1{\prolsymbol#1}
\def\@proldf{\@ifnextchar[{\@proldf@wrt}{\@proldf@}}
\def\@proldf@wrt[#1]#2{\@ifnextchar[{\@proldf@wrt@at{#1}{#2}}{\@proldf@wrt@{#1}{#2}}}
\def\@proldf@wrt@at#1#2[#3]{\prolsymbol^{*#1}_{#3}#2}
\def\@proldf@wrt@#1#2{\prolsymbol^{*#1}#2}
\def\@proldf@#1{\@ifnextchar[{\@proldf@at{#1}}{\@proldf@@{#1}}}
\def\@proldf@at#1[#2]{\prolsymbol^*_{#2}#1}
\def\@proldf@@#1{\prolsymbol^*#1}
\def\prolsymbol{\mathcal{T}}
\makeatother

\def\lcf{\lbrack\! \lbrack}
\def\rcf{\rbrack\! \rbrack}
\def\lcf{\lbrack\! \lbrack}
\def\rcf{\rbrack\! \rbrack}

\begin{document}

\title[ Some remarks on invariant Poisson quasi-Nijenhuis structures on Lie groups ]{Some remarks on invariant Poisson quasi-Nijenhuis structures on Lie groups}

\author[Gh. Haghighatdoost]{Ghorbanali Haghighatdoost}\address{Gh.\ Haghighatdoost:	Azarbaijan Shahid Madani University \\Department of Mathematics, Faculty of Science \\	Tabriz, Iran}\email{gorbanali@azaruniv.ac.ir}

\author[Z. Ravanpak]{Zohreh Ravanpak}\address{Z.\ Ravanpak:	Azarbaijan Shahid Madani University \\	Department of Mathematics, Faculty of Science \\Tabriz, Iran}\email{z.ravanpak@azaruniv.ac.ir }

\author[A. Rezaei-Aghdam]{Adel Rezaei-Aghdam}\address{A.\ Rezaei-Aghdam:	Azarbaijan Shahid Madani University \\	Department of Physics, Faculty of Science \\	Tabriz, Iran}\email{rezaei-a@azaruniv.ac.ir}


\thanks{This research was supported by research fund No. $217.d.14312$ from Azarbaijan Shahid Madani University.
}

\begin{abstract}
 We study {\em right-invariant (resp., left-invariant)  Poisson quasi-Nijenhuis structures} on a Lie group $G$ and introduce their infinitesimal counterpart, the so-called {\em r-qn structures} on the corresponding Lie algebra $\mathfrak g$. 
We investigate the procedure of the classification of such structures on the Lie algebras and then for clarity of our results we classify, up to a natural equivalence, all $r$-$qn$ structures on two types of four-dimensional real Lie algebras. We mention some remarks on the relation between $r$-$qn$ structures and the generalized complex structures on the Lie algebras $\mathfrak g$ and also the solutions of modified Yang-Baxter equation on the double of Lie bialgebra $\mathfrak g\oplus\mathfrak g^*$. The results are applied to some relevant examples.
\end{abstract}

\maketitle

\date{\today}

\section{Introduction}
Poisson quasi-Nijenhuis ($P$-$qN$) structures on manifolds were introduced by Sti\'enon and Xu \cite{StXu} as triples $(\Pi,\Phi,\mathbf N)$ on a manifold $M$ for which $\Pi$ is a Poisson $2$-vector, $\mathbf N$  is a $(1,1)$-tensor field and $\Phi$ is a closed $3$-form, such that $\Pi$ and  $\mathbf N$ are compatible in the sense of Poisson-Nijenhuis structures \cite{KoMa} and the Nijenhuis torsion of $\mathbf N$ is  
\[
[\mathbf N, \mathbf N](X,Y)=\Pi^{\sharp}(\iota_{X\wedge Y}\Phi), \quad \forall X,Y\in \frak{X}(M).
\]

In this work we study Poisson quasi-Nijenhuis structures on a Lie group $G$ which are appropriate right-invariant (or resp. left-invariant), so-called {\em right-invariant $P$-$qN$ structures} (or resp. {\em left-invariant $P$-$qN$ structures}) on $G$; we introduce their infinitesimal counterpart, the objects which we called {\em $r$-$qn$ structures} on the Lie algebra $\mathfrak g$ of $G$.

In fact Poisson-Nijenhuis structures are trivial Poisson quasi-Nijenhuis structures since for them $\Phi \equiv 0$. The infinitesimal counterpart of right-invariant $P$-$N$ structures, the structures which we called them $r$-$n$ structures, can be used to construct compatible solutions of classical Yang-Baxter equations (for more details see \cite{Zohreh}).

Since in the $r$-$qn$ structures, $(1,1)$-tensor field $\mathbf N$ is not Nijenhuis torsion free, in general the $2$-vector $nr$ is not an $r$-matrix. We show that in a certain condition we can have compatible $r$-matrices by $r$-$qn$ structures. 

In the following, we show that how we can obtain all $r$-$qn$ structures on the Lie algebra $\mathfrak g$ with finite dimension, or equivalently all right-invariant $P$-$qN$ structures on the connected simply-connected Lie group $G$ corresponding to $\mathfrak g$. Many of this structures would be equivalent by an Lie algebra automorphism, so in order to classify such structures we need to define an equivalence relation. We study the procedure of classification of $r$-$qn$ structures and classify, up to a natural equivalence, all $r$-$qn$ structures on two chosen real Lie algebras in dimension four, symplectic Lie algebra $A_{4,1}$ and non-symplectic Lie algebra $A_{4,8}$.
 
In \cite{StXu}, the authors studied generalized complex structures in terms of Poisson quasi-Nijenhuis manifold; they showed that a generalized complex manifold corresponds to a spacial class of Poisson quasi-Nijenhuis structures. So, it would be of interest to have $r$-$qn$ structures since whose spacial class corresponds to the generalized complex structures on the Lie algebra.

In fact in the infinitesimal level, where we deal with a Lie bialgebra, the generalized tangent bundle $\mathfrak g\oplus \mathfrak g^*$ of $\mathfrak g$ which is called the double of Lie bialgebra is equipped with a Lie algebra structure, so in this case a generalized complex structure on $\mathfrak g$ can be viewed as a spacial class of solutions of modified Yang-Baxter equation on the Lie algebra $\mathfrak g\oplus \mathfrak g^*$ .

We shall give some relevant examples of $r$-$qn$ structures on some Lie algebras $\mathfrak g$ which can be considered as a generalized complex structure on $\mathfrak g$ or as a solution of modified Yang-Baxter equation on $\mathfrak g\oplus \mathfrak g^*$.

The outline of the paper is as follows: In Section \ref{Section2} we briefly recall the notion of $P$-$qN$ structures on a manifold. We also take a review on cohomology of Lie algebra and then definition of Lie bialgebras, classical and modified Yang-Baxter equation. We end this section by the review of generalized complex strictures on a manifold. In Section \ref{Section3} we define right-invariant $P$-$qN$ structures on the Lie group $G$ as the main object of study and introduce their infinitesimal counterpart; the compatibility and equivalency of such structures will be also considered in this section.
In Section \ref{method} we describe the systematic way to get all $r$-$qn$ structures on $\mathfrak g$, equivalently all right-invariant $P$-$qN$ structures on $G$. The classification procedure of right-invariant $P$-$qN$ structures is the subject of Section \ref{cllasification}. we list the results of a classification of $r$-$qn$ structures  on two four dimensional Lie algebras, symplectic real Lie algebra $A_{4,1}$ and non-symplectic Lie algebra $A_{4,8}$; we explain all details in the procedure for Lie algebra $A_{4,1}$. We shall consider some remarks on $r$-$qn$ structures in Section \ref{app}; more precisely, we shall consider the conditions for which an $r$-$qn$ structure on Lie algebra $\mathfrak g$ defines a generalized complex structure on $\mathfrak g$ or an $R$-matrix on double of Lie algebra  $\mathfrak g \oplus\mathfrak g^*$. We end the paper by some relevant examples of $r$-$qn$ structures obtained in section \ref{cllasification}.

 \section{Antecedents }\label{Section2}
 In this section, we recall the definition of Poisson quasi-Nijenhuis
 structures \cite{StXu}. We will briefly review the notion of Lie bialgebra, classical and modified Yang-Baxter equation. We also take a review on generalized complex strictures on a manifold.

 \subsection{Poisson quasi-Nijenhuis structure}
 A Poisson quasi-Nijenhuis ($P$-$qN$)  structure on a manifold $M$ is a bivector field $\Pi: T^{*}M \times T^{*}M \to \mathbb R $, a $(1,1)$-tensor field ${\mathbf N}: TM \to TM$  together with a closed $3$-form $\Phi: TM \times TM \times TM \to \mathbb R $ on $M$ satisfying the conditions:
  \begin{enumerate}
  \item $\Pi$ is a Poisson bivector, i.e. $[\Pi, \Pi]\footnote{ $[\cdot, \cdot ]$ is the Schouten-Nijenhuis bracket.} = 0$,
  
  \item  
$
 [{\mathbf N}X,{\mathbf N}Y] - {\mathbf N}[{\mathbf N}X, Y] - {\mathbf N}[X,{\mathbf N}Y] + {\mathbf N}^2[X, Y]=\Pi ^{\sharp}(\Phi^{\sharp} (X,Y)), \quad \forall X, Y \in {\mathfrak X}(M),$

 \item
  ${\mathbf N} \Pi^{\sharp} = \Pi^{\sharp} {\mathbf N}^t,$ 
  
 \item
$ C(\Pi,{\mathbf N})(\alpha,\beta)=\left\lbrace \alpha,\beta\right\rbrace_{{\mathbf N}\Pi}- \left\lbrace \alpha,\beta\right\rbrace^{{\mathbf N}^{t}}_\Pi=0,\quad \forall \alpha, \beta \in {\Omega^1}(M),$

\item $d(i_{\mathbf N}\Phi)=0,$
\end{enumerate}
 where  $\Pi^{\sharp}: T^*M \to TM$ and $\Phi^{\sharp}:TM \times TM\to  T^*M$ are induced by $\Pi$ and $\Phi$, given by interior product,
 \[
\Pi^{\sharp} ( \alpha)=\iota_{\alpha}\Pi,\quad \Phi^{\sharp} ( X,Y)=\iota_{X\wedge Y}\Phi, \quad \alpha\in \Omega^1(M)\; \mbox{and} \; X,Y \in \mathfrak X(M).
 \]  
 ${\mathbf N}^{t}:T^{*}M \to T^{*}M$ is the dual $(1,1)$-tensor field to ${\mathbf N}$ and $C(\Pi, {\mathbf N})$ is a $(2,1)$-tensor field on $M$, a concomitant of $\Pi$ and ${\mathbf N}$, where  
\[
	\begin{array}{rcl}
\left\lbrace \alpha,\beta\right\rbrace^{{\mathbf N}^{t}}_\Pi
		&=&\left\lbrace {\mathbf N}^{t}\alpha,\beta\right\rbrace_{\Pi}+\left\lbrace \alpha,{\mathbf N}^{t}\beta \right\rbrace_{\Pi} -{\mathbf N}^{t}\left\lbrace \alpha,\beta \right\rbrace_{\Pi},\\[4pt]
		\end{array}
		\]
and the bracket $\{\cdot,\cdot\}_{\Pi}$ is the bracket of 1-forms which is defined by the Poisson bivector $\Pi$ as follows
\begin{equation}\label{Poisson-bracket}
\left\lbrace \alpha,\beta\right\rbrace _{\Pi}  =\mathcal L_{\Pi^{\sharp}\alpha}\beta-\mathcal L_{\Pi^{\sharp}\beta}\alpha+d(\Pi(\alpha,\beta)).
\end{equation}	
Similarly, the bracket $\left\lbrace \alpha,\beta\right\rbrace_{{\mathbf N}\Pi}$ is the bracket of 1-forms defined by the 2-contravariant tensor ${\mathbf N} \Pi$ (for more details see, \cite{KoMa}). 

\noindent Note that $i_{\mathbf N}$ is the derivation of degree $0$ defined by
\begin{equation}\label{in}
(i_{\mathbf N}\alpha)(X_1,...,X_p)=\sum_{i=1}^p\alpha(X_1,...,{\mathbf N}X_i,...,X_p),\quad \forall \alpha \in \Omega^p(M).
\end{equation}

{\bf Example}. Poisson-Nijenhuis structures on $M$ are trivial Poisson quasi-Nijenhuis, since for them the $3$-form $\Phi\equiv 0$.
   \subsection{Cohomology of Lie algebra } 
Let $\mathfrak g$ be a finite dimensional Lie algebra and $G$ be its corresponding simply connected Lie group. To each $\mathfrak g$-module $M$ and arbitrary nonnegative integer $k$ we can associate a {\em$k$-cochain} of $\mathfrak g$ with values in $M$ as a $k$-linear skew-symmetric map from $\mathfrak g$ to $M$. The $0$-cochain is just an element of $M$. Let us denote the space of $k$-cochains of $\mathfrak g$ with valuse in $M$ by $C^k(\mathfrak g,M)$.

 A linear map $\partial:C^k(\mathfrak g,M) \to C^{k+1}(\mathfrak g,M)$, satisfying $\partial^2=0$, is called a {\em coboundary operator}. We consider the definition of Chevalley-Eilenberg coboundary operator 
 \[
 \partial\delta(X_0,...X_k)=\sum _{i=0}^{k} (-1)^{i}X_i (\delta (X_0,...\hat{X_i}
 ,...,X_k))+\sum_{i<j}^{i+j}(-1)^{i+j}\delta([X_i,X_j],...,\hat {X_i},...,\hat{X_j},...,X_k),
 \]
 for $k$-cochain $\delta$ and $X_0,...,X_k \in \mathfrak g$. A $k$-cochain is called a {\em $k$-cocycle} if its coboundary is zero. 
 
Now we set $M:=\mathbb R$ with the trivial action of $\mathfrak g$ and use the abbreviate notation $C^k(\mathfrak g)$ instead of $C^k(\mathfrak g, \mathbb R)$. Note that, in this case the cohomology group is just the cohomology group of right-invariant (resp. left-invariant) forms on $G$ and the cobondary operator $\partial$ is exactly the de Rham differential $d$. The coboundary maps for $0,1,2$ and $3$-cochains $\theta, \eta,\mu$ and $\phi$ are:
\[
\begin{array}{rcl}
(\partial \theta)(X_i)& =& 0,\\
(\partial \eta)(X_i,X_j)& =& -\eta([X_i,X_j]),\\
(\partial \mu)(X_i,X_j,X_k)& =& -\mu([X_i,X_j],X_k)+\mu([X_i,X_k],X_j)-\mu([X_j,X_k],X_i),\\
\end{array}
\]
\begin{equation}\label{con3}
\begin{array}{rcl}
(\partial \phi)(X_i,X_j,X_k,X_l)& =& -\phi([X_i,X_j],X_k,X_l)+\phi([X_i,X_k],X_j,X_l)-\phi([X_i,X_l],X_j,X_k)\\
&&-\phi([X_j,X_k],X_i,X_l)+\phi([X_j,X_l],X_i,X_k)-\phi([X_k,X_l],X_i,X_j),\\
\end{array}
\end{equation}
 We remark that the cocycle condition is equivalent to the corresponding $k$-cochain being closed. So, a $k$-form on $\mathfrak g$ is closed if it is a $k$-cocycle in $C^k(\mathfrak g)$.

Now we can proceed to the definition of the Lie bialgebra.

   A {\em Lie bialgebra} $(\mathfrak g, \mathfrak g^*)$ is a Lie algebra with an additional structure, a linear map $\delta:\mathfrak g\to \mathfrak g\otimes \mathfrak g$ such that:
\begin{enumerate}
	\item[$(i)$] The linear map $\delta : \mathfrak{g}\longrightarrow \mathfrak{g}\otimes\mathfrak{g}$ is a $1$-cocycle, i.e.
	$$ad_X^{(2)}(\delta Y)-ad_Y^{(2)}(\delta X)- \delta[X,Y]=0,\quad \forall X,Y\in \mathfrak{g}, $$
	
where $ad_X^{(2)}$ is the adjoint representation of the Lie algebra $\mathfrak g$ on the space $\mathfrak g \otimes \mathfrak g$ defined by 

\[
ad_X^{(2)}(Y_1\otimes Y_2)=ad_{X}Y_1\otimes Y_2+Y_1\otimes ad_{X} Y_2\quad X,Y_1,Y_2 \in \mathfrak g.
\]

	\item[$(ii)$] The dual map $\delta^{t}: \mathfrak{g}^*\otimes{\mathfrak g}^* \longrightarrow \mathfrak g^{*}$ is a Lie bracket on $\mathfrak g^{*}$.
\end{enumerate}	
\noindent	We denote the Lie bracket on $\mathfrak g^*$ by $[\alpha,\beta]_*:=\delta^t(\alpha\otimes\beta)$ for $\alpha, \beta \in \mathfrak g^*$.

{\em Coboundary Lie bialgebra} is a Lie bialgebra defined by a $1$-cocycle $\partial r$ which is the coboundary of an element $r\in \mathfrak g\otimes \mathfrak g$ (for more details, see \cite{Ko}).

 We consider the following notation for the Sklaynin brackets on $\mathfrak g^*$ and $\mathfrak g$ defined by $r$-matrices $r$ and $r^{-1}$:
 \[
 [ X^i,  X^j]^r:=[ X^i,  X^j]_*={ \tilde f}^{ij}_{k} X^k, \quad \left[ X_i,X_j\right] ^{r^{-1}}:=[X_i,X_j]={f_{ij}}^{k} X_k,
 \]

\noindent where ${\tilde  f}^{ij}$ and $f_{ij}$ are structure constants of $\mathfrak g^*$ and $\mathfrak g$, respectively.
 
\subsubsection{Classical Yang-Baxter equation}
Let $\mathfrak g$ be finite-dimensional Lie algebra and $\mathfrak g^*$ be its dual vector space with respect to a non-degenerate canonical pairing $\left\langle\cdot,\cdot\right\rangle$, so for basis $\{X_i\}$ and dual basis $\{ X^i\}$ of $\mathfrak g$ and $\mathfrak g^*$, respectively, we have:
\[
\langle  X_i,X_j \rangle= \langle  X^i, X^j \rangle=0,
 \quad  \langle  X^i,X_j \rangle=\delta^i_j.
 \]
 To every element $r=r^{ij}X_i\otimes X_j$ of $\mathfrak g \otimes \mathfrak g$, we can associate the linear map $r^{\sharp}:{\mathfrak g}^* \to \mathfrak g$ defined by $r^{\sharp}({X}^i)({X}^j):=r({X}^i,{X}^j)$. Let $\delta_r:=\partial r$, then
 \[
 \delta_ r(X)=ad^{(2)}_{X}r=r^{ij}(ad_X X_i\otimes X_j+X_i\otimes ad_X X_j), \quad X\in \mathfrak g.
 \]
By definition, $\delta r$ is a $1$-cocycle. We denote the bracket on $\mathfrak g^*$ in this case by $[{X}^i,{X}^j]^r$ instead of $[{X}^i,{X}^j]_*$.

 \noindent To every element $r$ of $\mathfrak g \otimes \mathfrak g$ we can also associate a bilinear map $\left\langle r,r\right\rangle ^{\sharp}:{\mathfrak g}^* \times {\mathfrak g}^*\to \mathfrak g$ defined by
 \begin{equation}\label{con-r-matrix}
 \left\langle r,r\right\rangle ^{\sharp}({X}^i,{X}^j)=[r^{\sharp} {X}^i,r^{\sharp} {X}^j]-r^{\sharp}[{X}^i,{X}^j]^r,
 \end{equation}
which can be identified with an element $\langle r,r\rangle \in \wedge^2 \mathfrak g \otimes \mathfrak g$, such that
 \[
\langle r,r\rangle({X}^i,{X}^j,{X}^k):=\langle {X}^k,\left\langle r,r\right\rangle^{\sharp}({X}^i,{X}^j)\rangle.
 \]

\noindent For a skew-symmetric element $r \in \Lambda^2 {\mathfrak g} $, we have
 \begin{equation}\label{r-bracket}
[{X}^i,{X}^j]^r=ad^{*}_{r^{\sharp}{X}^i}{X}^j-ad^{*}_{r^{\sharp}{X}^j}{X}^i,
 \end{equation}
 where $ad^*_{X_i}=-(ad_{X_i})^t$ is the endomorphism of $\mathfrak g^*$  satisfying
 \begin{equation}
 \label{ad1}
 \langle  X^i, ad_{X_k}X_j\rangle =-\langle ad^{*}_{X_k} X^i,X_j \rangle,
 \end{equation}
 which implies
 \begin{equation}
  \label{ad2}
 \langle ad^{*}_{X_k} X^i, X_j\rangle =-\langle ad^{*}_{ X_j} X^i,X_k \rangle.
  \end{equation}

On the other hand, in the case where $r$ is skew-symmetric, we have $\left\langle  r,r \right\rangle =-\frac{1}{2}\lcf r,r\rcf$ where $\lcf \cdot,\cdot\rcf$ is Schouten-Nijenhuis bracket on the Lie algebra $\mathfrak g$ called the {\em algebraic Schouten bracket}.

For the skew-symmetric element $r\in \Lambda^2 {\mathfrak g}$, the condition $\lcf r,r\rcf=0$ is called the {\em classical Yang-Baxter equation} (CYBE). A solution of the CYBE is called an {\em $r$-matrix}. For any $r$-matrix the bracket (\ref{r-bracket}) is a Lie bracket on $\mathfrak g^*$, called the {\em Sklyanin bracket}. Therefore $r$-matrices can be identified with coboundary Lie bialgebras on the Lie algebra $\mathfrak g$ (for more details see, for instance, \cite{Ko}).
\begin{remark}
 $r\in \Lambda^2 {\mathfrak g}$ is a solution of classical Yang-Baxter equation if and only if 
 \begin{equation}\label{def-r-matrix}
r^{\sharp}[{X}^i,{X}^j]^r= [r^{\sharp} {X}^i,r^{\sharp} {X}^j], \quad \forall X^i,X^j\in \mathfrak g^*.
\end{equation}
 	\end{remark}
 	\subsubsection{Modified Yang-Baxter equation}\label{modified}
 	
Let $R$ be a linear map from finite-dimensional Lie algebra $\mathfrak g$ to itself. Consider the skew-symmetric bilinear form $\left\langle R,R\right\rangle _k$ on $\mathfrak g$ with values in $\mathfrak g$ defined by
\[
\left\langle R,R\right\rangle _k(X,Y)=[RX,RY]-R[RX,Y]-R[X,RY]+k[X,Y], \quad \forall X,Y \in \mathfrak g,\quad  k\in \mathbb R.
\] 

Condition $\left\langle R,R\right\rangle _k=0$ is called a {\em modified Yang-Baxter equation} (MYBE) with coefficient $k$. A solution of MYBE is called a {\em classical $R$-matrix} or $R$-matrix. 

\noindent We can define a bilinear skew-symmetric bracket $[\cdot,\cdot]_R$ on $\mathfrak g$ as
\[
[X,Y]_R=[RX,Y]+[X,RY],\quad \forall X,Y\in \mathfrak g.
\]
For $R$-matrix $R$, the above bracket defines a Lie algebra structure on $\mathfrak g$ which is called {\em double Lie algebra} (for more details see for example \cite{Ko}).

\subsection{Generalized complex structure}
For a manifold $M$, the space ${\mathcal T} M:=TM\oplus T^*M$ is called the {\em generalized tangent bundle} of $M$. The space of sections of ${\mathcal T} M$ is endowed with a bracket so called the {\em Courant bracket} given by
\[
[(X+\alpha),(Y+\beta)]=[X,Y]+{\mathcal L}_X\beta-{\mathcal L}_Y\alpha -\frac{1}{2}d\left(\alpha(Y)-\beta(x)\right)
\]

$\forall X,Y\in \mathfrak X(M), \forall \alpha, \beta \in \Omega ^1(M)$. This bracket is not a Lie bracket since it does not satisfy the Jacobi identity. The non-degenerate symmetric bilinear form on the vector space ${\mathcal T} M:=TM\oplus T^*M$ is defined by 
\[
\left\langle X+\alpha,Y+\beta\right\rangle =\alpha(X)+\beta(Y),\quad \forall X,Y\in \mathfrak X(M), \;\forall \alpha, \beta \in \Omega ^1(M),
\]
for more detail see \cite{Gu} and \cite {Hi}.
\begin{definition}\label{dif.J}
 A {\em generalized complex structure} on $M$ is a complex structures $\mathit J$ that is, a bundle map $\mathit J:{\mathcal T}M\to {\mathcal T}M$ which ${\mathit J}^2=-Id$ and $\left\langle \mathit Jv,\mathit Jw\right\rangle =\left\langle v,w\right\rangle $ satisfying the integrability condition
 \[
 [\mathit  Jv,\mathit  Jw]-\mathit  J [\mathit  Jv,w]-\mathit  J [v,\mathit  Jw]- [v,w]=0,\quad \forall v,w \in {\mathcal T}M.
 \] 
 \end{definition}
\begin{proposition}\label{g.c}
 A generalized complex structure $\mathit J$ on $M$ is of the form 
\begin{equation}\label{J}
\mathit{J}=\left(\begin{array}{cc}
	N& \Pi^{\sharp}\\ \Theta_{\sharp} & -N^t\end{array} \right),
\end{equation}

where $\Pi$ is a Poisson bivector on $M$ which is compatible with the vector bundle map $N:TM\to TM$ in the sense of Poisson-Nijenhuis structures, and $\Theta$ is a $2$-form on $M$ for which we denote the map $\Theta_{\sharp}:TM\to T^*M$ such that $\Theta_{\sharp}(X)=\iota _{X}\Theta$; satisfying the following conditions
\[
\begin{array}{rcl}
N^t\Theta_{\sharp}&=&\Theta_{\sharp}N\\[3pt]

N^2+\Pi^{\sharp}\Theta_{\sharp}&=&-Id\\[3pt]

[N,N](X,Y)&=&\Pi_{\sharp}\left( \iota_{X\wedge Y}(d\Theta)\right) \\[3pt]
d\Theta_{N}(X,Y,Z)&=& i_N(d\Theta)(X,Y,Z).\\
\end{array}
\]
For more detail see for example \cite{Cr} and \cite{StXu}.
\end{proposition}
Comparing the previous Proposition and the definition of Poisson-quasi Nijenhuis structures on a manifold we have the following Corollary.  
\begin{corollary}\label{cr.1}
The Poisson quasi-Nijenhuis structure $(\Pi, \Phi, N)$ on $M$ defines a generalized complex structure of the form (\ref{J}) on $M$ if $3$-form $\Phi$ is exact and the two following conditions hold

\begin{equation}\label{con.g.c}
\begin{array}{rcl}
N^t\Theta_{\sharp}&=&\Theta_{\sharp}N\\[3pt]

N^2+\Pi^{\sharp}\Theta_{\sharp}&=&-Id,\\[3pt]
\end{array}\end{equation}
where $\Theta$ is a $2$-form on $M$ such that $\Phi=d\Theta$.

\end{corollary}

\section{Right-invariant Poisson quasi-Nijenhuis structures} \label{Section3}
In this section we define right-invariant Poisson quasi-Nijenhuis ($P$-$qN$) structures on the Lie group $G$ and their infinitesimal counterpart on the Lie algebra $\mathfrak g$ of $G$. We also consider the concepts of compatibility and equivalence of those structures.

We are using the following notation. If $s$ is a $k$-vector on $\mathfrak g$ then $\rvec{s}$ (resp. $\lvec{s}$) is the right-invariant (resp. left-invariant) $k$-vector field on $G$ given by
\[
\rvec{s}(g) = (T_{\mathfrak e}R_g)(s), \; \; \mbox{ for } g \in G,
\]
(resp. $\lvec{s}(g) = (T_{\mathfrak e}L_g)(s)$, for $g \in G$) where $R_h: G \to G$ and $L_g: G \to G$ are the right and left translation by $h$ and $g$, respectively.

\begin{definition}
	A $P$-$qN$ structure $(\Pi,\Phi,{\mathbf N})$ on a Lie group $G$ is said to be right-invariant if:
	\begin{enumerate}
		\item
		The Poisson structure $\Pi$ is right-invariant, that is, there exists $r \in \Lambda^2 {\mathfrak g}$ such that $\Pi = \rvec{r}$.
		\item
		The closed $3$-form $\Phi$ is right-invariant, that is, there exist a real valued three linear, skew map $\phi\in C^3(\mathfrak g)$ satisfying $3$-cocycle condition, such that $\Phi = \rvec{\phi}$. 
		\item The $(1,1)$-tensor field ${\mathbf N}$ is right-invariant, that is, there exists a linear endomorphism $n: {\mathfrak g} \to {\mathfrak g}$ such that
		${\mathbf N} = \rvec{n}$.
	\end{enumerate}
\end{definition}
By the definition, we have 
\[
\begin{array}{rcl}
\Pi^{\sharp}_{(g)}&=& T _{\mathfrak e} R_{g} \circ r^{\sharp} \circ T^{*} _{\mathfrak e} R_{g},\\[4pt]
{\mathbf N}_{|T_gG} &=& T_{\mathfrak e}R_g \circ n \circ T_gR_{g^{-1}},\\[4pt]
\Phi^{{\sharp}^{\sharp}}_{(g)}&=& \wedge ^2(T^*_{g}R_{g^{-1}})\circ \phi^{{\sharp}^{\sharp}}\circ (T_gR_{g^{-1}}).\\
\end{array}
\]
 for $ g \in G$. For right-invariant $P$-$qN$ structures, we may prove the two following results which describe the infinitesimal version of such structures.

\begin{proposition}\label{inft-ver-right-1}
	Let  $(\Pi,\Phi,\mathbf N)$ be a right-invariant $P$-$qN$ structure on a Lie group $G$ with Lie algebra $\mathfrak g$ and identity element ${\mathfrak e}\in G$. If $r \in \Lambda^2{\mathfrak g}$ and $\phi \in\Lambda^3{\mathfrak g}^*$ are the value of $\Pi$ and $\Phi$ at ${\mathfrak e}$, and $n$ is the restriction of $N$ to ${\mathfrak g}$,  we have
	\begin{enumerate}
		\item[$(i)$] $r$ is a solution of the classical Yang-Baxter equation on ${\mathfrak g}$.
			\item[$(ii)$]
		The Nijenhuis torsion $	[n, n]$ of $n$ on ${\mathfrak g}$ equals
		\begin{equation}\label{con2}
		[n, n](X, Y) : = [nX, nY] - n[nX, Y] - n[X, nY] + n^2[X,Y] = r^{\sharp}(\phi^{\sharp}(X,Y)), \; \; \forall X, Y \in {\mathfrak g}.
		\end{equation}
		
			\item [$(iii)$] $\phi$ and $i_{n}\phi$ are $3$-cocycles with values in $\mathbb R$.
			\item[$(v)$] $n:\mathfrak g\longrightarrow \mathfrak g$ satisfies the condition
			\[
			n \circ r^{\sharp} = r^{\sharp} \circ n^t.
			\]
		\item[$(iv)$] The concomitant $C(r, n)$ of $r$ and $n$ in ${\mathfrak g}$ is zero, that is,
	\[
		C(r, n)(\alpha, \beta) : =
		\mathcal L^{\mathfrak g}_{r^{\sharp}\alpha}n^{t}\beta - \mathcal L^{\mathfrak g}_{r^{\sharp}\beta}n^{t}\alpha - n^{t}\mathcal L^{\mathfrak g}_{r^{\sharp}\alpha}\beta  + n^{t}\mathcal L^{\mathfrak g}_{r^{\sharp}\beta}\alpha=0\quad \mbox{ for } \alpha, \beta \in \mathfrak g^{*}.
	\]
		Here, ${\mathcal L}^{\mathfrak g}_{X}:=ad^*_{X}:\mathfrak g^* \to \mathfrak g^*$.
	\end{enumerate}
	
Conversely, let ${\mathfrak g}$ be a real Lie algebra of finite dimension, $r \in \Lambda^2{\mathfrak g}$ be a $2$-vector and $\phi \in \Lambda^3{\mathfrak g}^*$ be a $3$-form on ${\mathfrak g}$ and $n: {\mathfrak g} \to {\mathfrak g}$ be a linear endomorphism on ${\mathfrak g}$ which satisfy conditions $(i)$, $(ii)$, $(iii)$, $(iv)$ and (v); so-called $r$-$qn$ structure on the Lie algebra $\mathfrak g$. If $G$ is a Lie group with the Lie algebra ${\mathfrak g}$, then the triple $(\rvec{r}, \rvec {\phi},\rvec{n})$ is a right-invariant $P$-$qN$ structure on $G$.
\end{proposition}

\begin{proof}
The proofs of $(i)$, $(ii)$, $(v)$ and $(iv)$ are straightforward by the definition of right invariant objects (see \cite{Zohreh}). For $(iii)$, it is the consequence of the fact that a $k$-form on $\mathfrak g$ is closed if it is a $k$-cocycle in $C^k(\mathfrak g)$. 

\end{proof}	

We remark that, in case of $r$-$qn$ structures the 2-vector $nr$ is not an $r$-matrix since the operator $n$ is not a Nijenhuis operator. In the following Proposition we will see that under a certain condition it would be an $r$-matrix.

If there is not risk of confusion, we will use the same notation $r$ for the 2-vector $r$ and the linear map $r^{\sharp}$.


\begin{proposition}\label{con-2}
	Let $r\in \Lambda^2 \mathfrak g$ be an $r$-matrix and $n$ be an linear operator on $\mathfrak g$ which is compatible with $r$, that is $nr=rn^t$ and $C(r,n)\equiv 0$. Then, 2-vector $nr$ is an $r$-matrix if and only if
	\begin{equation}\label{con-rn}
	[n,n](rX^i,rX^j)=0,\quad \forall X^i,X^j\in {\mathfrak g}^*.
	\end{equation}	
\end{proposition}	

\begin{proof}
	2-vector $nr$ defines a bracke $[X^i,X^j]^{nr}=ad^*_{nrX^i}X^j-ad^*_{nrX^j}X^i$ on $\mathfrak g$. On the other hand  $C(r,n)\equiv 0$ implies
\[
	[X^i,X^j]^{nr}=[n^tX^i,X^j]^{r}+[X^i,n^tX^j]^{r}-n^t[X^i,X^j]^r.
\]
Equivalently, $C(r,n)\equiv 0$ guaranties $2$-vectors $r$ and $nr$ are compatible, that is $\lcf r,nr \rcf \equiv 0$, so we have
	\begin{equation}\label{con-1}
r[  X^i, X^j] ^{nr}+nr[  X^i,  X^j] ^{r}=[ rX^i,nr X^j] +[nr X^i,r X^j].
\end{equation}
Suppose the condition (\ref{con-rn}) holds, that is
	\begin{eqnarray}\label{nr-q1}
	[nrX^i,nrX^j]-n[nrX^i,rX^j]-n[rX^i,nrX^j]+n^2[rX^i,rX^j]=0.
	\end{eqnarray}	

Using (\ref{def-r-matrix}) for $r$-matrix $r$ and comparing two relations (\ref{con-1}) and (\ref{nr-q1}), we get	
	\[
	nr[X^i,X^j]^{nr}=[nrX^i,nrX^j]
	\]
which from \ref{def-r-matrix} means $nr$ is an $r$-matrix. One proves the converse in a similar way.
	
\end{proof}		

\begin{corollary}
In the case that $r$-matrix $r$ in the previous Proposition is non-degenerate, $nr$ is an $r$-matrix if and only if $n$ is a Nijenhuis operator.
	
\end{corollary}	

	
 \subsection{Compatibility of right-invariant Poisson quasi-Nijenhuis structures} \label{compatible}
 Two $P$-$qN$ structures $(\Pi,\Phi, \mathbf N)$ and $(\Pi',\Phi', \mathbf N')$ on a Lie group $G$ are said to be {\em compatible} if the couple $(\Pi+\Pi',\Phi+\Phi', \mathbf N+\mathbf N')$ is a $P$-$qN$ structure on $G$.

In the case of the right-invariant $P$-$qN$ structures, the compatibility of structures reduces to the compatibility of their infinitesimal version.

 If $(r,\phi,n)$ and $(r',\phi',n')$ be the infinitesimal versions of the mentioned $P$-$qN$ structures, then they are {\em compatible} if the couple $(r+r',\phi+\phi',n'+n)$ is an $r$-$qn$ structure on the Lie algebra $\mathfrak g$ of $G$, that is $n+n'$ and $r+r'$ are compatible and
 \[
 \lcf r,r' \rcf=0, \quad [n,n']\footnote{ [n,n'] is the Nijenhuis concomitant (\cite{Ni},\cite{Zohreh}) of two $(1,1)$-tensor fields  $n'$ and $n $ on the Lie algebra $\mathfrak g$ defined by
 	\[
 \begin{array}{rcl}
 [n,n'](X_i,X_j)& =& [n,n](X_i,X_j)+[n',n'](X_i,X_j)+[ n'X_i,nX_j] +[nX_i,n'X_j] - n'[nX_i, X_j] \\[4pt]
 && - n'[X_i,nX_j]- n[n'X_i, X_j] - n[X_i,n'X_j] + n\circ n'[X_i,X_j] + n'\circ n[X_i,X_j],
 \end{array}
\]
for every two elements of basis $\{X_i\}$ of $\mathfrak g$.}	
 	 =r\phi^{\sharp}+r\phi'^{\sharp}+r'\phi^{\sharp}+r'\phi'^{\sharp}.
 \]

 From the Proposition \ref{inft-ver-right-1}, $(r+r',\phi+\phi',n+n')$ would be the infinitesimal version of the right-invariant $P$-$qN$ structure $(\rvec {r+r'},\rvec {\phi+\phi'},\rvec{n+n'})$ on the Lie group $G$.

  \subsection{Equivalence classes of right-invariant Poisson quasi-Nijenhuis structures}\label{equivalence.}
In \cite{Zohreh}, we defined the equivalence class of right-invariant Poisson-Nijenhuis structures. Now, we define the equivalence classes of right-invariant $P$-$qN$ structures. Two right-invariant $P$-$qN$ structures $(P,\Phi,\mathbf N)$ and $(P',\Phi',\mathbf N)$ on the Lie group $G$ are {\em equivalent} if two corresponding $r$-$qn$ structures $(r,\phi,n)$ and $(r',\phi',n')$ on the Lie algebra $\mathfrak g$ are equivalent.

\begin{definition}\label{equivalnt}
	Two $r$-$qn$ structures $(r,\phi,n)$ and $(r',\phi',n')$ are equivalent if there exist a Lie algebra automorphism $\mathcal A$ such that the following diagrams commute,
	\[
	\begin{tikzpicture}
	\matrix (m) [matrix of math nodes, row sep=3em, column sep=4em]
	{  & \mathfrak g  & \mathfrak g^*  & \mathfrak g &\mathfrak g  \\
		& \mathfrak g  & \mathfrak g^* & \mathfrak g & \mathfrak g \\ };
	{ [start chain] 
		\chainin (m-1-2);
		{ [start branch=A] \chainin (m-2-2)
			[join={node[left,labeled] {\mathcal A}}];}
		\chainin (m-1-3) [join={node[above,labeled] {\phi_i}}];
		{ [start branch=B] \chainin (m-2-3)
			[join={node[right,labeled] {\mathcal A^{-t}}}];}
		\chainin (m-1-4) [join={node[above,labeled] {r}}];
		{ [start branch=C] \chainin (m-2-4)
			[join={node[right,labeled] {\mathcal A}}];}
		\chainin (m-1-5)[join={node[above,labeled] {n}}];
		{ [start branch=C] \chainin (m-2-5)
			[join={node[right,labeled] {\mathcal A}}];}
	}
	{ [start chain] 
		\chainin (m-2-2);
		\chainin (m-2-3) [join={node[below,labeled] {\phi'_i}}];
		\chainin (m-2-4) [join={node[below,labeled] {r'}}];
		\chainin (m-2-5) [join={node[below,labeled] {n'}}]; }
		\end{tikzpicture}
		\]
that is, 
\[
r\sim_\mathcal{A} r', \quad n\sim_\mathcal{A} n',\quad \phi_i\sim_\mathcal{A} \phi'_i,\quad \mbox{for}\; i:=1,...,dim \mathfrak g,
\] 
where we define $\phi_i=\phi^{\sharp^{\sharp}}(X_i)\in \Lambda^2{\mathfrak g}^*$. In fact the map $\phi^{\sharp}:\mathfrak g\times \mathfrak g \to {\mathfrak g}^*$ induces a map  $\phi^{\sharp^{\sharp}}:\mathfrak g \to \Lambda^2{\mathfrak g}^*$ defined by
\[
\phi^{\sharp^{\sharp}}(X_i)(X_j,X_k)=\phi^{\sharp}(X_i,X_j)(X_k)=\phi(X_i,X_j,X_k),\quad \{X_i\}\in \mathfrak g
\]
which can be interpreted by the map $\phi_i:\mathfrak g\to {\mathfrak g}^*$ for every $i:=1,...,dim \mathfrak g$.

We will write $(r,\phi,n)\sim (r',\phi',n')$ ($(r,\phi,n)\sim_\mathcal{A}(r',\phi',n')$ if we want to indicate $\mathcal{A}$).
	\end{definition}

\rm

 \section{$r$-$qn$ structures on Lie algebras}\label{method}
 In this section we describe how we can get all $r$-$qn$ structures on $\mathfrak g$, equivalently right-invariant $P$-$qN$ structures on $G$. For this purpose we rewrite the five conditions of Proposition \ref{inft-ver-right-1} in terms of coordinates. Throughout this section we denote the basis $\{X_i\}$ and the dual basis $\{ X^j\}$ for Lie algebras $\mathfrak g$ and $\mathfrak g^*$, respectively.

First, we write the structural constants $f^{k}_{ij}$ of the Lie algebra $\mathfrak g$, in terms of adjoint representation ${\mathcal {X}_{i}}$, and antisymmetric matrices ${\mathcal Y^i}$, as
 \begin{equation}\label{Matrix}
 f^{k}_{ij}=-(\mathcal{Y}^{k})_{ij},\quad f^{k}_{ij}=-(\mathcal {X}_{i})_{j}^{~k}.
 \end{equation}

\noindent{\bf Condition (i)} Consider the tensor notation of the CYBE{\footnote {$
		\lcf r,r \rcf =[r_{12},r_{13}]+[r_{12},r_{23}]+[r_{13},r_{23}]=0,
		$,  where $
		r_{12}=r^{ij}X_i\otimes X_j\otimes 1$, $r_{13}=r^{ij}X_i\otimes 1\otimes X_j$ and $r_{23}=r^{ij} 1\otimes X_i\otimes X_j$. }}, (see \cite{Ko}). We can rewrite condition $\lcf r,r\rcf\equiv 0$ in the matrix form (see \cite{Zohreh})
\begin{equation}\label{Poisson}
  r{\mathcal Y}^ {i}r-r{\mathcal X}_{l}r^{il}-r^{il}{\mathcal X}_{l}^{t}r=0,\quad i:=1,...,dim \mathfrak g.
\end{equation}

  \noindent{\bf Condition (ii)}  We write the condition (\ref{con2}) for two base elements ${X_i}$ and ${X_j}$ in $\mathfrak g$ and we get
 \[
n^{k}_{~i}n^{l}_{~j}f^{m}_{kl}- n^{k}_{~i}n^{m}_{~l}f^{l}_{kj}-n^{k}_{~j}f^{l}_{ik}n^{m}_{~l}+f^{k}_{ij}n^{l}_{~k}n^{m}_{~l}=\phi_{ijk}r^{kl}.
 \]
 By using (\ref{Matrix}) it can be rewritten in the matrix form
 \begin{equation}\label{Nijenhuis}
-n^l
_{~i}n^{t}{\mathcal X}_{l}+n^{l}_{~i}{\mathcal X}_{l}n^{t}- {\mathcal X} _in^{t}n^{t}+n^{t}{\mathcal X}_{i}n^{t}=\phi_{i}r,\quad i:=1,...,dim \mathfrak g,
 \end{equation}
 where $\phi_i:=\phi(X_i)=\phi_{ijk}X^{j}\wedge X^{k}$.

   \noindent{\bf Condition (iii)} The cocycle condition (\ref{con3}) for $\phi$ in the base elements is
   \[
   \partial \phi(X_i,X_j,X_k,X_l)=-f_{ij}^{m}\phi_{mkl}+f_{ik}^{m}\phi_{mjl}-f_{il}^{m}\phi_{mjk}-f_{jk}^{m}\phi_{mil}+f_{jl}^{m}\phi_{mik}-f_{kl}^{m}\phi_{mij}=0,
    \]
   which can be rewritten in the matrix form by using (\ref{Matrix}) as follows
   \begin{equation}\label{closed1}
   {\mathcal X}_{j}\phi_{k}-{\mathcal X}_{k}\phi_{j}+{\mathcal Y}^{m}(\phi_{m})_{jk}+({\mathcal Y}^{m})_{jk}\phi_{m}+\phi_{k}{\mathcal X}^{t}_{j}-\phi_{j}{\mathcal X}^{t}_{k}=0,\quad j,k:=1,...,dim \mathfrak g,
   \end{equation}
   where elements of the matrix $\phi_{i}$ are $(\phi_{i})_{jk}=\phi_{ijk}$.
   
\noindent Using the relation (\ref{in}) for $\phi \in\Lambda^3{\mathfrak g}^*$ and endomorphism $n$ we have
\[
(i_{n}\phi)(X_i,X_j,X_k)=\phi(nX_i,X_j,X_k)+\phi(X_i,nX_j,X_k)+\phi(X_i,X_j,nX_k).
\]
Using (\ref{con3}), the cocycle condition for $(i_{n}\phi)$ is
\[
\begin{array}{rcl}
 \partial (i_{n}\phi)(X_i,X_j,X_k,X_l)= &-& f_{ij}^{m}(n_{m}^{~s}\phi_{skl}+n_{k}^{~s}\phi_{msl}+n_{l}^{~s}\phi_{mks})+f_{ik}^{m}(n_{m}^{~s}\phi_{sjl}+n_{j}^{~s}\phi_{msl}+n_{l}^{~s}\phi_{mjs})\\[3pt]
  &-& f_{il}^{m}(n_{m}^{~s}\phi_{sjk}+n_{j}^{~s}\phi_{msk}+n_{k}^{~s}\phi_{mjs})-f_{jk}^{m}(n_{m}^{~s}\phi_{sil}+n_{i}^{~s}\phi_{msl}+n_{l}^{~s}\phi_{mis})\\[3pt]
    &+& f_{jl}^{m}(n_{m}^{~s}\phi_{sik}+n_{i}^{~s}\phi_{msk}+n_{k}^{~s}\phi_{mis})-f_{kl}^{m}(n_{m}^{~s}\phi_{sij}+n_{i}^{~s}\phi_{msj}+n_{j}^{~s}\phi_{mis})\\
    &=& 0
\end{array}
\] 
\noindent and then, using (\ref{Matrix}), we get the matrix relation
  \begin{equation}\label{closed2}
\begin{array}{rcl}
&~&{\mathcal X}_{j}n\phi_{k}+{\mathcal X}_{j}\phi_{s}n_{k}^{~s}+{\mathcal X}_{j}\phi_{k}n^{t}-{\mathcal X}_{k}n\phi_{j}-{\mathcal X}_{k}\phi_{s}n_{j}^{~s}-{\mathcal X}_{k}\phi_{j}n^{t}-{\mathcal Y}^{m}(n\phi_{j})_{mk}+\\[3pt]
&&{\mathcal Y}^{m}(n\phi_{m})_{jk}-{\mathcal Y}^{m}(n\phi_{m})_{kj}+({\mathcal X}_{j}n)_{k}^{~s}\phi_{s}+({\mathcal Y}^{m})_{jk}n\phi_{m}+({\mathcal Y}^{m})_{jk}\phi_{m}n^t+\\[3pt]
&& n\phi_{k}{\mathcal X}_{j}^{t}+n\phi_{k}{\mathcal X}_{j}^{t}+n_{k}^{~s}\phi_{s}{\mathcal X}_{j}^{t}-n\phi_{j}{\mathcal X}_{k}^{t}-n\phi_{j}{\mathcal X}_{k}^{t}-n_{j}^{~s}\phi_{s}{\mathcal X}_{k}^{t}=0\\
\end{array}
\end{equation}
   
  \noindent{\bf Condition (iv)} For every element $X^{i}$ in $\mathfrak g^*$, $n\circ r^{\sharp}(X^i)=r^{\sharp}\circ n^t(X^i)$ implies
  \begin{equation}\label{Con1}
  (nr)^i_{~j}=(rn^t)^i_{~j}, \quad i,j=1,...,dim(\mathfrak g),
  \end{equation}
  where $n$ and $r$ are the corresponding matrices to the linear operator $n$ and the map $r^{\sharp}$ and $nX_i=n_i^jX_j$ for $n_i^j\in \mathbb R$.

 \noindent{\bf Condition (v)} By applying $ X^{i}$ and $X^{j}$ in the concomitant and then, using (\ref{Matrix}), we get the matrix relation (see \cite{Zohreh})
 \begin{equation}\label{Con2}
 r{\mathcal X}_{j}n^{j}_{~i}+{\mathcal X}^{t}_{j}n^{j}_{~i}r-r{\mathcal X}_{i}n^{t}-n{\mathcal X}^{t}_{i}r=0,\quad i:=1,...,dim \mathfrak g.
 \end{equation}
Given a Lie algebra $\mathfrak g$ with finite dimension, by applying matrices ${\mathcal X}_i$ and ${\mathcal Y}^i$ using (\ref{Matrix}), in six relations (\ref{Poisson}), (\ref{Nijenhuis}), (\ref{closed1}),(\ref{closed2}), (\ref{Con1}), (\ref{Con2}) and solving them by help of mathematical softwares, one can find all $r$-$qn$ structures on $\mathfrak g$ and so, all right-invariant $P$-$qN$ structures on the Lie group $G$.

 \section{classification procedure}\label{cllasification}

Many of $r$-$qn$ structures obtained in section \ref{method} would be equivalent by an  Lie algebra automorphism. In this section we will proceed in five step to show how we can classify, up to an equivalence, all $r$-$qn$ structures on a Lie algebra. For clarity of results, we exemplify the procedure by classifying all $r$-$qn$ structures on two types of four dimensional real Lie algebras, symplectic real Lie algebra $A_{4,1}$ and non-symplectic Lie algebra $A_{4,8}$. We explain all details of classification procedure for Lie algebra $A_{4,1}$ \footnote{We use the notations of the four-dimensional real Lie algebras denoted in \cite{Abedi}, (see also \cite{Christ}).}. We did all computations using Maple.\\ 


\noindent The strategy is as follows.

\noindent {\bf First step.} using (\ref{Poisson}) we find all $r$-matrices on four-dimensional symplectic real Lie algebra $\mathfrak g$ and classify them up to equivalence 
\[
r\sim r'\quad \Leftrightarrow \quad\exists \mathcal{A}\in Aut(\mathfrak{g}) \quad \mathcal{A}r\mathcal{A}^t=r'.
\]

\noindent  {\bf Second step.} we take a representative of each class of $r$-matrices in first step, and find all endomorphisms $n$ on $\mathfrak g$ which are compatible with the chosen $r$-matrix by solving relations (\ref{Con1}) and (\ref{Con2}); they give five equations on four-dimensions. Then we classify all obtained pairs $(r,,n)$ up to equivalence                                   
\begin{equation}\label{e1}(r,n')\sim_0 (r,n) \quad \Leftrightarrow \quad\exists \mathcal{A}\in Aut(\mathfrak{g}) \quad\mathcal{A}r\mathcal{A}^t=r \ \&\ \mathcal{A}n\mathcal{A}^{-1}=n'\, ,
\end{equation}
where $\sim _{0}$ indicates the equivalence for the couple ($r$,$n$) with the same $r$. 

 \noindent{\bf Third step.} Now we find all $3$-forms $\phi$ which satisfy in the relation (\ref{Nijenhuis}) for $(1,1)$-tensor fields $n$ we found in the second step. In order to, we write the skew symmetric maps $\phi_i$ in the matrix forms. In dimension four they are as follow
\begin{equation}\label{phis}
\begin{array}{rcl}
\phi_1=\left(\begin{array}{cccc}
	0& 0& 0 & 0\\ 0& 0& \phi_{123} & \phi_{124}\\ 0& -\phi_{123}& 0 & \phi_{134}\\ 0& -\phi_{124}& -\phi_{134} & 0
	\end{array} \right)\quad &&	\phi_2=\left(\begin{array}{cccc}
	0& 0& -\phi_{123} & -\phi_{124}\\ 0& 0& 0 & 0 \\ \phi_{123}& 0& 0 & \phi_{234}\\ \phi_{124}& 0& -\phi_{234} & 0
	\end{array} \right)\\
&\\
	\phi_3=\left(\begin{array}{cccc}
	0& \phi_{123}& 0 & -\phi_{134}\\ -\phi_{123}& 0& 0 & -\phi_{234} \\ 0& 0& 0 & 0\\ \phi_{134}& \phi_{234}& 0 & 0
	\end{array} \right)\quad &&	\phi_4=\left(\begin{array}{cccc}
	0&\phi_{124} & \phi_{134} & 0\\ -\phi_{124}& 0& \phi_{234} & 0 \\ -\phi_{134}& -\phi_{234}& 0 & 0\\ 0& 0& 0 & 0
	\end{array}\right).\\
\end{array}
\end{equation}             
 {\bf Fourth step}: In this step we check if $(1,1)$-tensor fields $n$ from the second step and $3$-forms $\phi$ from the third step satisfy in relations (\ref{closed1}) and (\ref{closed2}), or they may have some new conditions.

\noindent {\bf Fifth step.} Finally, we classify all obtained pairs $(r,\phi,n)$ up to equivalence 
	\begin{equation}\label{e1}(r,\phi,n)\sim_0 (r,\phi',n) \quad \Leftrightarrow \quad\exists \mathcal{A}\in Aut(\mathfrak{g}) \quad\mathcal{A}r\mathcal{A}^t=r \ \&\ \mathcal{A}n\mathcal{A}^{-1}=n\, \ \&\ \mathcal{A}^{t}\phi_i\mathcal{A}=\phi'_i\,  ,
		\end{equation}
for $i:=1,...,dim \mathfrak g$.	Here $\sim _{0}$ indicates the equivalence for $r$-$qn$ structures with the same $r$ and $n$.

	\begin{proposition} \label{Pro}
		If $\{r_\alpha\}_\Lambda$ is a set of all representatives of the equivalence relation $r\sim r'$ and $\{(r_\alpha,n_\beta)\}_{(\alpha,\beta)\in\Psi}$ is a set of all representatives of the equivalence relations
		$(r_\alpha,n)\sim_0(r_\alpha,n')$ and  $\{(r_\alpha,({\phi_i})_{\eta},n_\beta)\}_{(\alpha,\eta,\beta)\in\Gamma}$ is a set of all representatives of the equivalence relations
		$(r_\alpha,({\phi_i})_{\eta},n_\beta)\sim_0(r_\alpha,({\phi'_i})_{\eta},n_\beta)$, then $\{(r_\alpha,(\phi_i)_{\eta},n_\beta)\}_{(\alpha,\eta,\beta)\in\Gamma}$ is a set of representatives of the equivalence relation $(r,\phi_i,n)\sim(r',\phi'_i,n')$ (c.f. Definition \ref{equivalnt}).
	\end{proposition}
	\begin{proof}
		Consider an $r$-$qn$ structure $(r,\phi,n)$. There exist $r_\alpha$ and $\mathcal{A}$ such that $r_\alpha=\mathcal{A}r\mathcal{A}^t$. Take the representative $(r_\alpha,n_\beta)$ of $(r_\alpha, \mathcal{A}n\mathcal{A}^{-1})$. Then $(r_\alpha,n_\beta)$ represents the class of $(r,n)$ 
		under $\sim$. Now take the representative $(r_\alpha,(\phi_i)_\eta,n_\beta)$ of $(r_\alpha, \mathcal{A}^t\phi_i\mathcal{A},n_\beta)$. Then, It is easy to see that $(r_\alpha,(\phi_i)_{\eta},n_\beta)$ represents the class of $(r,\phi_i,n)$ 
		under $\sim$.
		
		Moreover, different elements of $\{(r_\alpha,(\phi_i)_{\eta},n_\beta)\}_{(\alpha,\eta,\beta)\in\Gamma}$ represent different elements of $\sim$. Indeed, if $(r_\alpha,(\phi_i)_{\eta},n_\beta)\sim(r_{\alpha'},(\phi_i)_{\eta'},n_{\beta'})$, then there is $\mathcal{A}$ such that
			\[
			r_\alpha=\mathcal{A}r_{\alpha'}\mathcal{A}^t \quad \& \quad n_\beta=\mathcal{A}n_{\beta'}\mathcal{A}^{-1}\quad \& \quad (\phi_i)_\eta=\mathcal{A}^t(\phi_i)_{\eta'}\mathcal{A}.
			\]
	\noindent But then, by definition, $\alpha=\alpha'$ and hence $(r_\alpha,n_\beta)\sim_0(r_\alpha,n_{\beta'})$, thus $\beta=\beta'$; which means different elements of $\{(r_\alpha,n_\beta)\}_{(\alpha,\beta)\in\Psi}$ represent different elements of $\sim$. Therefore for $(r_\alpha,(\phi_i)_\eta ,n_\beta)\sim_0(r_\alpha,(\phi_i)_{\eta'},n_{\beta})$, implies $\eta=\eta'$. 
		
			\end{proof}
\noindent In the following we clarify the above procedure by describing the details for the Lie algebra $A_{4,1}$. We list the results in any step for Lie algebras $A_{4,1}$ and $A_{4,8}$.

\subsection{$r$-matrices} We consider the Lie algebra $A_{4,1}$ with non-zero commutators $[X_2,X_4]=X_1$ and $[X_3,X_4]=X_2$, from (\ref{Matrix}) we have the following matrices $\mathcal{X}_i$ and $\mathcal{Y}^i$
\begin{equation}\label{XY}
\begin{array}{rclcrcl}
 	\mathcal{X}_1=\left(\begin{array}{cccc}
 	0& 0& 0 & 0\\ 0& 0& 0 & 0\\ 0& 0& 0 & 0\\ 0& 0& 0 & 0
 	\end{array} \right)&& 	\mathcal{X}_2=\left(\begin{array}{cccc}
 0& 0& 0 & 0\\ 0& 0& 0 & 0 \\ 0& 0& 0 & 0\\ -1& 0& 0 & 0
 	\end{array} \right)\\
 &\\
	\mathcal{X}_3=\left(\begin{array}{cccc}
 	0& 0& 0 & 0\\ 0& 0& 0 & 0 \\ 0& 0& 0 & 0\\ 0& -1& 0 & 0
 	\end{array} \right) &&
\mathcal{X}_4=\left(\begin{array}{cccc}
 	0& 0& 0 & 0\\ 1& 0& 0 & 0 \\ 0& 1& 0 & 0\\ 0& 0& 0 & 0
 	\end{array}\right) \\
 &\\
	\mathcal{Y}^1=\left(\begin{array}{cccc}
 	0& 0& 0 & 0\\ 0& 0& 0 & -1\\ 0& 0& 0 & 0\\ 0& 1& 0 & 0
 	\end{array} \right) &&
\mathcal{Y}^2=\left(\begin{array}{cccc}
 	0& 0& 0 & 0\\ 0& 0& 0 & 0 \\ 0& 0& 0 & -1\\ 0& 0& 1 & 0
 	\end{array} \right) \\
 &\\
	\mathcal{Y}^3=\left(\begin{array}{cccc}
 	0& 0& 0 & 0\\ 0& 0& 0 & 0 \\ 0& 0& 0 & 0\\ 0& 0& 0 & 0
 	\end{array} \right) &&
 	\mathcal{Y}^4=\left(\begin{array}{cccc}
 	0& 0& 0 & 0\\ 0& 0& 0 & 0 \\ 0& 0& 0 & 0\\ 0& 0& 0 & 0
 	\end{array} \right).\\
 \end{array}
 \end{equation}

We list the classification of $r$-matrices for Lie algebra $A_{4,1}$ and $A_{4,8}$ in the table 1. Note that we classified all such structures on four dimensional symplectic real Lie algebras in \cite{Zohreh}. For self containing of the paper we bring the result of Lie algebra $A_{4,1}$ (see table 1 in \cite{Zohreh}).\\
\begin{center}
	{\footnotesize   \bf{Table 1.}} \label{TT}
	{\footnotesize Classification of $r$-matrices on four-dimensional real Lie algebra $A_{4,1}$ and $A_{4,8}$. }\\   \begin{tabular}{| l | l l l | p{15mm} }
		
		\hline\hline
		{\scriptsize  $A_{4,1}$}   &{\scriptsize  Equivalence classes of $r$-matrices}
		& {\scriptsize  $$}& {\scriptsize  $$}
		\smallskip\\
		\hline
		\smallskip
		{\scriptsize  $f_{24}^1=1$}& 
		{\scriptsize  $\ast \; c^{12}X_1\wedge X_2+c^{13}X_1 \wedge X_3+c^{14}X_1\wedge X_4+c^{23}X_2 \wedge X_3$} &  {\scriptsize$c^{12},c^{13}\in \mathbb R, $} &  {\scriptsize$c^{14},c^{23}\in \mathbb R-\{0\}.$}\\[3pt]
		
		{\scriptsize  $f_{34}^2=1$} &{\scriptsize  $c^{12}X_1\wedge X_2+c^{13}X_1 \wedge X_3+c^{23}X_2\wedge X_3$}&   
		{\scriptsize  $c^{12},c^{13}\in \mathbb R,  $}&
		{\scriptsize  $c^{23}\in \mathbb R- \{0\}.$}\\[3pt]
		
		{\scriptsize  $$}	&
		{\scriptsize  $c^{12}X_1\wedge X_2+c^{13}X_1 \wedge X_3+c^{14}X_1\wedge X_4 $} &
		{\scriptsize  $c^{12},c^{13}\in \mathbb R,  $}&
		{\scriptsize  $c^{14}\in \mathbb R- \{0\}.$}\\[3pt]
		
		{\scriptsize  $$} &{\scriptsize  $c^{12}X_1\wedge X_2+c^{13}X_1 \wedge X_3$}&  {\scriptsize  $c^{12}\in \mathbb R,$}  &
		{\scriptsize  $c^{13}\in  \mathbb R^{+}-\{0\}.$} \\[3pt]
		
		{\scriptsize  $$}  &{\scriptsize  $c^{12}X_1\wedge X_2+c^{13}X_1 \wedge X_3$}&  {\scriptsize  $c^{12}\in \mathbb R,$}  &
		{\scriptsize  $c^{13}\in  \mathbb R^{-}-\{0\}.$} \\[3pt]
		
		{\scriptsize  $$} &{\scriptsize  $c^{12}X_1\wedge X_2$} &
		{\scriptsize  $$}&  {\scriptsize  $c^{12}\in \mathbb R-\{0\}.$}  \\

		\hline
		{\scriptsize  $A_{4,8}$}   &{\scriptsize  Equivalence classes of $r$-matrices}
		& {\scriptsize  $$}& {\scriptsize  $$}
		\smallskip\\
		\hline
		\smallskip
		{\scriptsize  $f_{23}^1=1$}& 	 
		{\scriptsize  $c^{12}X_1\wedge X_2+c^{23}(\displaystyle \frac{c^{23}}{c^{24}}X_1 \wedge X_3+X_2 \wedge X_3+X_1\wedge X_4)+c^{24}X_2\wedge X_4$} &  {\scriptsize$c^{12},c^{23}\in \mathbb R, $} &  {\scriptsize$c^{24},\in \mathbb R-\{0\}.$}\\[3pt]
		
		{\scriptsize  $f_{24}^2=1$}& 	 
		{\scriptsize  $c^{13}X_1\wedge X_3+c^{23}(\displaystyle \frac{c^{23}}{c^{34}}X_1 \wedge X_2+X_2\wedge X_3-X_1 \wedge X_4)+c^{34}X_3\wedge X_4$} &  {\scriptsize$c^{12},c^{23}\in \mathbb R, $} &  {\scriptsize$c^{34},\in \mathbb R-\{0\}.$}\\[3pt]
		
		{\scriptsize  $f_{34}^3=-1$}& 	 
		{\scriptsize  $c^{12}X_1\wedge X_2+c^{13}X_1 \wedge X_3+c^{14}X_1\wedge X_4$} &  {\scriptsize$c^{12},c^{13}\in \mathbb R, $} &  {\scriptsize$c^{14},\in \mathbb R-\{0\}.$}\\[3pt]
		
		{\scriptsize  $$}& 	 
		{\scriptsize  $c^{12}X_1\wedge X_2+c^{13}X_1 \wedge X_3$} &  {\scriptsize$ $} &  {\scriptsize$c^{12},c^{13}\in \mathbb R-\{0\}.$}\\[3pt]
		
		{\scriptsize  $$}& 	 
		{\scriptsize  $c^{12}X_1\wedge X_2+c^{13}X_1 \wedge X_3$} &  {\scriptsize$c^{13}\in \mathbb R, $} &  {\scriptsize$c^{12},\in \mathbb R-\{0\}.$}\\[3pt]
		
		{\scriptsize  $$}& 	 
		{\scriptsize  $c^{12}X_1\wedge X_2+c^{13}X_1 \wedge X_3$} &  {\scriptsize$c^{12}\in \mathbb R, $} &  {\scriptsize$c^{13},\in \mathbb R-\{0\}.$}\\[3pt]
		
		\hline
	\end{tabular}
\end{center}
\vspace{3mm} 
 \subsection{$r$-$qn$ structures }
 We take a representative on each class
\[
\begin{array}{rclcrclcrcl}
&& r_0=X_1\wedge X_4-X_2\wedge X_3, && r_1=X_1\wedge X_2-X_2\wedge X_3,&&  r_2=X_1\wedge X_2-X_1\wedge X_3+X_1\wedge X_4,\\[3pt]

&& r_3=X_1\wedge X_2+X_1\wedge X_3,&& r_4=X_1\wedge X_2-X_1\wedge X_3,&& r_5=X_1\wedge X_2.
\end{array}
\]

Now, we find all $r$-$qn$ structures on Lie algebra $A_{4,1}$ for each $r$-matrices $r_i$ $(i:=0,...,5)$.

It is easy to see that every $3$-form on this Lie algebra is close or equivalently is a $3$-cocycle, in fact using (\ref{con3}), we have
\[
\partial\phi(X_1,X_2,X_3,X_4)=\phi(X_1,X_1,X_3)-\phi(X_2,X_1,X_2)=0.
\]
The same we see that $i_n\phi$ is also a $3$-cosycle, since
\[
\partial(i_n\phi)(X_1,X_2,X_3,X_4)=i_n\phi(X_1,X_1,X_3)-i_n\phi(X_2,X_1,X_2),
\]
which vanishes apart from whatever $i_n\phi$ is. Therefore, it only needs to be solved the equations (\ref{Nijenhuis}), (\ref{Con1}) and (\ref{Con2}).

Take a generic $(1,1)$-tensor field $n=\sum_{i,j=1}^{4}n^i_{j}X_i\otimes X^j$. By inserting $\mathcal{X}_i$ and $\mathcal{Y}^i$ (\ref{XY}), the matrix forms $n$ in the relations  (\ref{Con1}) and (\ref{Con2}) we find all $(1,1)$-tensor fields $n_i$ which are compatible with $r_i$. For example for the two first $r$-matrices we find the compatible couples $(r_0,n_{r_0})$ and $(r_1,n_{r_1})$ with the following $(1,1)$-tensor fields,
\[
\begin{array}{rcl}
	n_{r_0}&=&\left(\begin{array}{cccc}
	n_1& -n_2& n_4& 0\\ 0& n_3& 0 & n_4\\ 0& 0& n_3 & n_2\\ 0& 0& 0 & n_1
	\end{array} \right),\\
	&&\\
	n_{r_1}&=&\left(\begin{array}{cccc}
	n_1+n_2& 0& n_3-n_2& n_4\\ 0& n_1+n_3& 0 & n_5\\ n_1& 0& n_3 & n_6\\ 0& 0& 0 & n_2
	\end{array} \right).
	\end{array}
\]

\noindent We indicate by $n_{r_i}$, the $(1,1)$-tensor fields compatible with $r_i$; and for simplicity, we used the notation $n_i$ for the element of matrices $n_{r_i}$ instead of $n_i^j$.

Finally, by inserting the matrices (\ref{XY}) of $\mathcal{X}_i$, $\mathcal{Y}^i$ , matrix forms (\ref{phis}) of $\phi_i$ and the matrix forms of $r_i$ and $n_{r_i}$ for $(i:=1,...,4)$, in the relation (\ref{Nijenhuis}) and solving equations we will find $\phi_i$'s and thus $3$-forms $\phi$ on $\mathfrak g$.

We find $r$-$qn$ structures $(r_0,\phi_{r_0},n_{r_0})$ and $(r_1,\phi_{r_1},n_{r_1})$ where 
\[
\phi_{r_0}\equiv 0,\quad \phi_{r_1}=(n_1^2+n_1n_3-n_1n_2)X_1\wedge X_3\wedge X_4.
\]
So, there is no non-trivial $r$-$qn$ structure with $r_0$ on this Lie algebra. About $r_1$, we impose the following conditions in order to $\phi_i$'s are non-zero..
\begin{equation}\label{con-phi}
n_1\neq 0,\quad \mbox{ and} \quad n_1\neq n_2- n_3.
 \end{equation}
\\
All $r$-$qn$ structures on four-dimensional symplectic real Lie algebra $A_{4,1}$ and non-symplectic real Lie algebra $A_{4,8}$ are given in tables $2.a$ and $2.b$, respectively. Note that, the first column gives the non-vanishing structural constants of the Lie algebra $\mathfrak g^*$ defined by the corresponded $r$-matrix in column two. \\
\newpage
\begin{center}
{\footnotesize   \bf{Table 2.a.}} 
{\footnotesize $r$-$qn$ structures on four-dimensional symplectic real Lie algebras $A_{4,1}$ . }\\   \begin{tabular}{ | l | l l |p{40mm} }
	
	\hline\hline
{\scriptsize  $ {\tilde f}^{ij}_{k}$}&{\scriptsize  $\mbox{ $r$-matrix $r$}$}
	& {\scriptsize  $\mbox{$(1,1)$-tensor field $n$} $}\\
{\scriptsize  $$}&{\scriptsize  $\mbox{3-form $\phi$} $}
	& {\scriptsize  $$}
	\smallskip\\
	\hline
	\smallskip
	
 {\scriptsize ${ \tilde f}^{12}_3=1$}&
	{\scriptsize  $r=X_1 \wedge X_4-X_2 \wedge X_3$} &  {\scriptsize$n(X_1)=n_1X_1$} \\
	
 {\scriptsize ${\tilde f}^{13}_4=1$}&
		{\scriptsize  $\phi\equiv 0$} &  {\scriptsize$ n(X_2)=-n_2X_1+n_3X_2,$} \\
	
{\scriptsize  $$}& {\scriptsize  $ $}  &   {\scriptsize  $n(X_3)=n_4X_1+n_3X_3$}  \\

{\scriptsize  $$}& {\scriptsize  $ $}  &   {\scriptsize  $n(X_4)=n_4X_2+n_2X_3+n_1X_4$}  \\
	\hline
	 {\scriptsize ${ \tilde f}^{12}_3=1$}&
	 {\scriptsize  $r=X_1 \wedge X_2+X_1 \wedge X_4-X_1 \wedge X_3$} &  {\scriptsize$n(X_1)=(n_2-n_1)X_1$} \\
	 
	 {\scriptsize  ${\tilde f}^{12}_4=1$}& {\scriptsize  $ \phi=-(n_1^2+n_1n_6)X^2\wedge X^3\wedge X^4$}  &   {\scriptsize  $n(X_2)=(n_3-n_4)X_1+(n_2+n_6)X_2$}  \\
	 
	 {\scriptsize  $$}& {\scriptsize  $ $}  &   {\scriptsize  $n(X_3)=n_3X_1+(n_1+n_5+n_6)X_2$}  \\ 
	 
	 {\scriptsize  $$}& {\scriptsize  $ $}  &   {\scriptsize  $\quad\quad\quad\quad +(n_2-n_1+n_6)X_3+n_1X_4$}  \\

	 {\scriptsize  $$}& {\scriptsize  $ $}  &   {\scriptsize  $ n(X_4)=n_4X_1+n_5X_2+n_6X_3+n_2X_4$}  \\ 
	 \hline
	 
 {\scriptsize ${ \tilde f}^{12}_4=1$}&
	{\scriptsize  $r=X_1 \wedge X_2-X_1 \wedge X_3$} &  {\scriptsize$n(X_1)=n_1X_1$} \\
	
{\scriptsize  $$}& {\scriptsize  $ \phi=(2n_1n_4-n_1^2-n_4^2)X^2\wedge X^3\wedge X^4$}  &   {\scriptsize  $n(X_2)=n_2X_1+n_4X_2$}  \\

{\scriptsize  $$}& {\scriptsize  $ $}  &   {\scriptsize  $n(X_3)=n_2X_1+(n_4-n_1)X_2+n_1X_3$}  \\

{\scriptsize  $$}& {\scriptsize  $ $}  &   {\scriptsize  $ n(X_4)=n_3X_1+n_5X_2+n_6X_3+n_4X_4$}  \\ 
	\hline
 {\scriptsize ${ \tilde f}^{13}_4=1$}&
	{\scriptsize  $r=X_1 \wedge X_2-X_2 \wedge X_3$} &  {\scriptsize$n(X_1)=(n_1+n_2)X_1+n_1X_3$} \\
	
{\scriptsize  $$}& {\scriptsize  $ \phi=(n_1^2+n_1n_3-n_1n_2)X^1\wedge X^3\wedge X^4$}  &   {\scriptsize  $n(X_2)=(n_1+n_3)X_2$}  \\
	
{\scriptsize  $$}& {\scriptsize  $ $}  &   {\scriptsize  $n(X_3)=(n_3-n_2)X_1+n_3X_3$}  \\

{\scriptsize  $$}& {\scriptsize  $ $}  &   {\scriptsize  $ n(X_4)=n_4X_1+n_5X_2+n_6X_3+n_2X_4$}  \\ 
	\hline
		
 {\scriptsize ${ \tilde f}^{ij}_k=0$}  & 	{\scriptsize  $r=X_1 \wedge X_2$} &{\scriptsize$n(X_1)=n_1X_1$}\\
		
{\scriptsize  $$}	  &{\scriptsize  $ \phi=(n_1^2-n_1n_6-n_1n_9-n_7n_8$}  & {\scriptsize  $n(X_2)=n_1X_2$}   \\
		
	{\scriptsize  $$}  	&  {\scriptsize  $~~~~+n_6n_9)X^1\wedge X^3\wedge X^4+(n_1n_4-n_4n_9 $}& {\scriptsize  $n(X_3)=n_2X_1+n_4X_2+n_6X_3+n_8X_4$} \\

{\scriptsize  $$} 	& {\scriptsize  $~~~~+n_5n_8)X^2\wedge X^3\wedge X^4 $}   & {\scriptsize  $ n(X_4)=n_3X_1+n_5X_2+n_7X_3+n_9X_4$}\\ 
		
		\hline
	\end{tabular}
\end{center}

	\begin{center}
{\footnotesize   \bf{Table 2.b.}} \label{TT}
{\footnotesize $r$-$qn$ structures on four-dimensional non-symplectic real Lie algebras  $A_{4,8}$ . }\\   \begin{tabular}{ | l | l l |p{40mm} }
	
	\hline\hline	
		{\scriptsize  $ {\tilde f}^{ij}_{k}$}&{\scriptsize  $\mbox{ $r$-matrix $r$}$}
		& {\scriptsize  $\mbox{$(1,1)$-tensor field $n$} $}\\
		{\scriptsize  $$}&{\scriptsize  $\mbox{3-form $\phi$} $}
		& {\scriptsize  $$}
		\smallskip\\
		\hline
		\smallskip
			
			{\scriptsize ${ \tilde f}^{23}_3=-1$}&
			{\scriptsize  $r=X_2 \wedge X_4$} &  {\scriptsize$n(X_1)=n_1X_1$} \\[3pt]
			
			{\scriptsize ${\tilde f}^{14}_3=-1$}&
			{\scriptsize  $\phi\equiv 0$} &  {\scriptsize$ n(X_2)=n_2X_2$} \\[3pt]
			
			{\scriptsize  ${\tilde f}^{24}_4=-1$}& {\scriptsize  $ $}  &   {\scriptsize  $n(X_3)=n_3X_2+n_1X_3+n_4X_4$}  \\[3pt]
			{\scriptsize  $$}& {\scriptsize  $ $}  &   {\scriptsize  $n(X_4)=n_2X_4$}  \\[3pt]
			
			\hline
			
			{\scriptsize ${ \tilde f}^{23}_2=-1$}&
			{\scriptsize  $r=X_3 \wedge X_4$} &  {\scriptsize$n(X_1)=n_1X_1$} \\[3pt]
			
			{\scriptsize ${\tilde f}^{14}_2=1$}&
			{\scriptsize  $\phi\equiv 0$} &  {\scriptsize$ n(X_2)=n_1X_2+n_2X_3+n_4X_4$} \\[3pt]
			
			{\scriptsize  $$}& {\scriptsize  $ $}  &   {\scriptsize  $n(X_3)=n_3X_3$}  \\[3pt]
			{\scriptsize  $$}& {\scriptsize  $ $}  &   {\scriptsize  $n(X_4)=n_3X_4$}  \\[3pt]
			
			\hline
			
			{\scriptsize ${ \tilde f}^{12}_2=1$}&
			{\scriptsize  $r=X_1 \wedge X_2+X_1 \wedge X_3+X_1 \wedge X_4$} &  {\scriptsize$n(X_1)=(n_1+n_2)X_1$} \\[3pt]
			
			{\scriptsize ${\tilde f}^{13}_3=-1$}&
			{\scriptsize  $\phi= -(n_2n_5)X^2 \wedge X^3\wedge X^4$} &  {\scriptsize$ n(X_2)=-(n_3+n_4)X_1+(n_1+n_2-n_5)X_2$} \\[3pt]
			
			{\scriptsize  ${ \tilde f}^{12}_4=-1$}& {\scriptsize  $ $}  &   {\scriptsize  $n(X_3)=n_3X_1+n_1X_3$}  \\[3pt]
			{\scriptsize  ${ \tilde f}^{13}_4=1$}& {\scriptsize  $ $}  &   {\scriptsize  $n(X_4)=n_4X_1+n_5X_2+n_2X_3+(n_1+n_2)X_4$}  \\[3pt]
			
			\hline
				{\scriptsize ${ \tilde f}^{12}_4=-1$}&
				{\scriptsize  $r=X_1 \wedge X_2+X_1 \wedge X_3$} &  {\scriptsize$n(X_1)=(2n_1-n_2)X_1$} \\[3pt]
				
				{\scriptsize ${\tilde f}^{13}_4=1$}&
				{\scriptsize  $\phi= (2n_1^2-4n_1n_2+2n_2^2)X^1 \wedge X^2 \wedge X^4$} &  {\scriptsize$ n(X_2)=-n_3X_1+n_1X_2+(n_1-n_2)X_3$} \\[3pt]
				
				{\scriptsize  $$}& {\scriptsize  $\quad \quad+(4n_1n_2-2n_1^2-2n_2^2)X^1 \wedge X^3 \wedge X^4 $}  &   {\scriptsize  $n(X_3)=n_3X_1+(n_1-n_2)X_2+n_1X_3$}  \\[3pt]
				
				{\scriptsize  $$}& {\scriptsize  $ \quad \quad+(n_2-n_1)(n_5+n_6)X^2 \wedge X^3 \wedge X^4$}  &   {\scriptsize  $n(X_4)=n_4X_1+n_5X_2+n_6X_3+n_2X_4$}  \\[3pt]
				
				\hline
		\end{tabular}
		\end{center}
				\begin{center}
					{\footnotesize   \bf{Table 2.b.}} \label{TT}
					{\footnotesize $r$-$qn$ structures on four-dimensional non-symplectic real Lie algebras  $A_{4,8}$ . }\\   \begin{tabular}{ | l | l l |p{40mm} }
						
						\hline\hline	
						{\scriptsize  $ {\tilde f}^{ij}_{k}$}&{\scriptsize  $\mbox{ $r$-matrix $r$}$}
						& {\scriptsize  $\mbox{$(1,1)$-tensor field $n$} $}\\
						{\scriptsize  $$}&{\scriptsize  $\mbox{3-form $\phi$} $}
						& {\scriptsize  $$}
						\smallskip\\
						\hline
						\smallskip
			
			{\scriptsize ${ \tilde f}^{12}_4=-1$}&
			{\scriptsize  $r=X_1 \wedge X_2$} &  {\scriptsize$n(X_1)=n_1X_1$} \\[3pt]
			
			{\scriptsize $$}&
			{\scriptsize  $\phi= (n_5n_6-n_4n_7-n_1n_5)X^2 \wedge X^3 \wedge X^4$} &  {\scriptsize$ n(X_2)=n_1X_2$} \\[3pt]
			
			{\scriptsize  $$}& {\scriptsize  $ $}  &   {\scriptsize  $n(X_3)=n_2X_1+n_4X_2+n_6X_3$}  \\[3pt]
			
			{\scriptsize  $$}& {\scriptsize  $ $}  &   {\scriptsize  $n(X_4)=n_3X_1+n_5X_2+n_7X_3+n_1X_4$}  \\[3pt]
			
			\hline
			
			{\scriptsize ${ \tilde f}^{13}_4=1$}&
			{\scriptsize  $r=X_1 \wedge X_3$} &  {\scriptsize$n(X_1)=n_1X_1$} \\[3pt]
			
			{\scriptsize $$}&
			{\scriptsize  $\phi= (n_4n_7-n_5n_6-n_1n_7)X^2 \wedge X^3 \wedge X^4$} &  {\scriptsize$ n(X_2)=n_2X_1+n_4X_2+n_6X_3$} \\[3pt]
			
			{\scriptsize  $$}& {\scriptsize  $ $}  &   {\scriptsize  $n(X_3)=n_1X_3$}  \\[3pt]
			
			{\scriptsize  $$}& {\scriptsize  $ $}  &   {\scriptsize  $n(X_4)=n_3X_1+n_5X_2+n_7X_3+n_1X_4$}  \\[3pt]
			
			\hline
	\end{tabular}
\end{center}	
\subsection{Equivalence classes of $r$-$qn$ structures }	
\noindent We use the following automorphism group element (classified in \cite{Christ}, see also \cite{ReSe1}) of Lie algebra $A_{4,1}$
\[
	\mathcal A=\left(\begin{array}{cccc}
	a_{11} a_{16}^2 & a_7 a_{16} & a_3 & a_4\\ 0 & a_{11} a_{16}& a_7 & a_8\\0 & 0 & a_{11} & a_{12}\\ 0 & 0 & 0 & a_{16}
	\end{array} \right).
\]

Since for $r_0$ we have $\phi\equiv 0$, it means $(1,1)$-tensor $n$ is a Nijenhuis operator and in fact the couple $(r_0,n_{r_0})$ is an $r$-$n$ structure and these structures classified in \cite{Zohreh}, (see the table 3 of \cite{Zohreh}).

\noindent For the $r$-matrix $r_1=X_1\wedge X_2-X_2\wedge X_3$, we insert the above $\mathcal A$ in the relation $\mathcal A\circ r_1-r_1\circ \mathcal A^{-t}=0$ and we get
\begin{equation}\label{Ar1}
a_{16}=\frac{1}{(a_{11})^2},\quad a_3=\frac{(a_{11})^4-1}{(a_{11})^3},\quad a_7=0,
\end{equation}
The automorphism group $\mathcal A$ in the new expression, given by (\ref{Ar1}) is
\begin{equation}\label{Ar11}
	\mathcal A=\left(\begin{array}{cccc}
	\frac{1}{(a_{11})^3} & 0 & \frac{(a_{11})^4-1}{(a_{11})^3} & a_4\\ 0 & \frac{1}{a_{11}}& 0 & a_8\\0 & 0 & a_{11} & a_{12}\\ 0 & 0 & 0 & \frac{1}{(a_{11})^2}
	\end{array} \right).
\end{equation}
Since $det(\mathcal A)=\displaystyle \frac{1}{(a_{11})^5}$, $a_{11}\neq 0$ and other parameters $a_4,a_8,a_{12}$ can take any value.

\noindent Now, we find all equivalence classes of $(1,1)$-tensor fields $n_{r_1}$ such that $(r_1,n'_{r_1})\sim_0 (r_1,n_{r_1})$, where
\[
n'_{r_1}=\left(\begin{array}{cccc}
	n'_1+n'_2& 0& n'_3-n'_2& n'_4\\ 0& n'_1+n'_3& 0 & n'_5\\ n'_1& 0& n'_3 & n'_6\\ 0& 0& 0 & n'_2
	\end{array} \right).
\]
Therefore, we have all equivalence classes of the couples $(r_1,n^i_{r_1})$ corresponding to the $r$-matrix $r_1$. In order to do, we insert the automorphism group (\ref{Ar11}) in the relation $n_{r_1}\circ \mathcal A-\mathcal A\circ n'_{r_1}=0$, we obtain the following equations

\begin{itemize}

\item[(1)]
$n_2-n'_2=0,\quad  n_3-n'_3=0,\quad  (n_1+n_2)-(n'_1+n'_2)=0, \quad (n_1+n_3)-(n'_1+n'_3)=0$,

\item[(2)]
$  (a_{11})^4(n_3-n'_1-n'_3)+n'_1+n'_2-n_2=0,$ 
 
 \item[(3)]
 $a_8(a_{11})^2(n_2-n'_1-n'_3)+a_{11}n_5-n'_5=0,\quad (a_{11})^2(a_{11}n_6+a_{12}n_2-a_{12}n'_3)-n'_6=0,$
 
 \item [(4)]
 $(a_{11})^3(n_6a_{11}+a_4n_2-a_4n'_1-a_4n'_2-a_{12}n'_3+a_{12}n'_2)-a_{11}n'_4+n_4-n_6=0.$
\end{itemize}

The set of equations $(1)$ implies that $n_1:=n'_1$, $n_2:=n'_2$ and $n_3:=n'_3$ which means three parameters $n_1$, $n_2$ and $n_3$ are free. Note that, we mean by free parameters $n_i$ , the parameters for which different values get non-equivalent Nijenhuis structures belonging to the different equivalence classes. The equation $(2)$ implies that $a_{11}=\pm 1$. Applying $a_{11}=\pm 1$ in the equations $(3)$ and $(4)$ we get
\[
\begin{array}{rcl}
n_6\pm n'_6&=& \pm a_{12}(n_3-n_2),\\
n_5\pm n'_5& =& \pm a_8(n_1+n_3-n_2),\\
n_4\pm n'_4& =& \pm \left(a_{12}(n_3-n_2)+a_4n_1\right).\\
\end{array}
\]
From the first equation, if $n_2=n_3$ then $n_6=\pm n'_6$; it means that for the structures whose $n_2=n_3$ the parameter $n_6$ is free; note that in this case the structures with the opposite sign of $n_6$ are equivalent. For the structure whose $n_2\neq n_3$ the parameter $n_6$ can be any constant because of arbitrary parameter $a_{12}$. From the second and third equations, parameters $n_4$ and $n_5$ can be any arbitrary constant since parameters $a_8$ and $a_{12}$ are arbitrary. We indicate arbitrary constants $n_i$ by $c_i$. We mean by arbitrary constants $c_i$, the parameters such that for every different values of them, the corresponding Nijenhuis structures are equivalent belonging to the same class.

\noindent Finally, we get the following  equivalence classes of $(1,1)$-tensor $n$
\[
\begin{array}{rcl}
n_{r_1}^{(1)}&=&\left(\begin{array}{cccc}
	n_1+n_2& 0& n_3-n_2& c_4\\ 0& n_1+n_3& 0 & c_5\\ n_1& 0& n_3 & c_6\\ 0& 0& 0 & n_2
	\end{array} \right),\quad n_2\neq n_3\neq 0,\\
	&&\\
n_{r_1}^{(2)}&=&\left(\begin{array}{cccc}
	n_1+n_2& 0& n_3-n_2& c_4\\ 0& n_1+n_3& 0 & c_5\\ n_1& 0& n_3 & n_6\\ 0& 0& 0 & n_2	\end{array} \right),\quad n_6 \in {\mathbb R}^+,\; n_2=n_3.\\
\end{array}
\]
 Note that, all parameters are in $\mathbb R$ unless we mention some conditions for them. Therefore, we get the equivalence classes of couples $(r_1,n_{r_1}^{(1)})$ and $(r_1,n_{r_1}^{(2)})$ with additional conditions (\ref{con-phi}).

Now we consider the equivalence class of $3$-form $\phi$. We should solve the equations
 \begin{equation}\label{equ-phi}
 {\mathcal A}^t\phi_i \mathcal A-\phi'_i=0,\quad  \mbox{for} \quad i:=1,...,4,
 \end{equation}
where $\phi'_i$'s are the same as $\phi_i$ defined in (\ref{phis}), the difference is that we denote the elements of  $\phi'_i$'s by $\phi'_{ijk}$ instead of $\phi_{ijk}$, and the automorphism group $\mathcal A$ is (\ref{Ar11}) with the condition $a_{11}=\pm 1$.

For $\phi_1$, the equation (\ref{equ-phi}) gets $n_1^2+n_1n_3-n_1n_2=0$ which can not be zero, so we have no more equivalence class of $\phi$.
	
One may check the equation (\ref{equ-phi}) for $\phi_1$, $\phi_2$ and $\phi_3$ and get the same results, but actually it dose not need to be done because according to the Definition \ref{equivalnt}, $\phi \sim \phi'$ if and only if $\phi_i \sim \phi'_i$ for all $i=1,...,4$.

Eventually, we have the triples $(r_1,\phi_{r_1},n^{(1)}_{r_1})	$ and  $(r_1,\phi_{r_1},n^{(2)}_{r_1})	$ as equivalence classes of $r$-$qn$ structures on Lie algebra $A_{4,1}$ with $r$-matrix $r_1$. 

All equivalence classes of $r$-$qn$ structures on four-dimensional symplectic real Lie algebra $A_{4,1}$ and non-symplectic real Lie algebra $A_{4,8}$ are listed in tables $3.a$ and $3.b$, respectively.    

Note that we use the following Lie algebra automorphism for Lie algebra $A_{4,8}$ in the procedure of classification
\[
	\mathcal A=\left(\begin{array}{cccc}
	a_{11} a_{16} & a_{12}a_6 & a_{11}a_{8} & a_4\\ 0 &  a_{6}& 0 & a_8\\0 & 0 & a_{11} & a_{12}\\ 0 & 0 & 0 & 1
	\end{array} \right).
\]

\newpage
\begin{center}
{\footnotesize   \bf{Table 3.a.}} \label{TT}
{\small Equivalence classes of $r$-$qn$ structures on four-dimensional symplectic real Lie algebra $A_{4,1}$}

  \begin{tabular}{ | l  l | l | p{14mm} }
\hline\hline

{\scriptsize  $r$-matrix $r$}
&{\scriptsize $(1,1)$-tensor field $n$ }	&{\scriptsize $Comments$ }\\
 {\scriptsize  $ $}
&{\scriptsize $3$-form $\phi$ }	&{\scriptsize $$ }\\
\hline
\smallskip
	{\scriptsize  $r=X_1 \wedge X_2-X_1 \wedge X_3$} &{\scriptsize$n(X_1)=n_1X_1$}  &  	{\scriptsize $n_1\neq n_4$}\\

 {\scriptsize  $ \phi=(2n_1n_4-n_1^2-n_4^2)X^2\wedge X^3\wedge X^4$}  & {\scriptsize  $n(X_2)=n_2X_1+n_4X_2$}    &	{\scriptsize  $ n_1\; \mbox{or}\; n_4\neq 0$}\\

	{\scriptsize  $$}& {\scriptsize  $n(X_3)=n_2X_1-n_2X_2+n_1X_3$}  &{\scriptsize  $$}  \\

{\scriptsize  $$} &{\scriptsize  $ n(X_4)=c_3X_1+n_5X_2+n_6X_3+n_4X_4$}&	{\scriptsize  $$}   \\ 
\hline

	{\scriptsize  $r=X_1 \wedge X_2-X_1 \wedge X_3$}& {\scriptsize$n(X_1)=n_1X_1$} & {\scriptsize  $n_1\neq n_4$}	\\

{\scriptsize  $ \phi=(2n_1n_4-n_1^2-n_4^2)X^2\wedge X^3\wedge X^4$}  & {\scriptsize  $n(X_2)=n_2X_1+n_4X_2$}    &{\scriptsize  $ n_1\; \mbox{or}\; n_4\neq 0$}	\\

	{\scriptsize  $$}  & {\scriptsize  $n(X_3)=n_2X_1+(n_4-n_1)X_2+n_1X_3$} & {\scriptsize  $n_2\neq n_1-n_4 $}\\

{\scriptsize  $$} &{\scriptsize  $ n(X_4)=c_3X_1+n_5X_2+n_6X_3+n_4X_4$} &{\scriptsize  $$}  \\ 
\hline


	{\scriptsize  $r=X_1 \wedge X_2-X_2 \wedge X_3$}& {\scriptsize$n(X_1)=(n_1+n_2)X_1+n_1X_3$}  & {\scriptsize  $n_2\neq n_3 $}	\\

{\scriptsize  $ \phi=(n_1^2+n_1n_3-n_1n_2)X^1\wedge X^3\wedge X^4$} &  {\scriptsize  $n(X_2)=(n_1+n_3)X_2$}     & {\scriptsize  $ n_1\neq n_2-n_3$}\\

 	{\scriptsize  $$}  &{\scriptsize  $n(X_3)=(n_3-n_2)X_1+n_3X_3$}&{\scriptsize  $$}  \\

{\scriptsize  $$}  &{\scriptsize  $ n(X_4)=c_4X_1+c_5X_2+c_6X_3+n_2X_4$}&{\scriptsize  $$}  \\ 
\hline
 	{\scriptsize  $r=X_1 \wedge X_2-X_2 \wedge X_3$}& {\scriptsize$n(X_1)=(n_1+n_2)X_1+n_1X_3$}   & {\scriptsize  $n_1\neq 0 $}\\

{\scriptsize  $ \phi=(n_1^2+n_1n_3-n_1n_2)X^1\wedge X^3\wedge X^4$}& {\scriptsize  $n(X_2)=(n_1+n_2)X_2$}      &	{\scriptsize  $n_6\in \mathbb R^+ $} \\

	{\scriptsize  $$}  &{\scriptsize  $n(X_3)=n_2X_3$}& {\scriptsize  $$}  \\

{\scriptsize  $$} & {\scriptsize  $ n(X_4)=c_4X_1+c_5X_2+n_6X_3+n_2X_4$}&{\scriptsize  $$}  \\ 
\hline

	{\scriptsize  $r=X_1 \wedge X_2$} &{\scriptsize$n(X_1)=n_1X_1+n_2X_3+n_3X_4$}  &  {\scriptsize  $n_5n_8\neq 0 $}	\\

 {\scriptsize  $ \phi=-(n_7n_8)X^1\wedge X^3\wedge X^4$}  &{\scriptsize  $n(X_2)=n_1X_2+n_4X_3+n_5X_4$}    &	 {\scriptsize  $n_7n_8\neq 0  $}\\

	{\scriptsize  $~~~~+(n_5n_8)X^2\wedge X^3\wedge X^4 $}  &{\scriptsize  $n(X_3)=n_6X_3+n_7X_4$}& {\scriptsize  $$}  \\

	{\scriptsize  $$}  &{\scriptsize  $ n(X_4)=n_8X_3+n_1X_4$}&{\scriptsize  $$}  \\ 
\hline
	{\scriptsize  $r=X_1 \wedge X_2$} & {\scriptsize$n(X_1)=n_1X_1+n_2X_3+n_3X_4$} & {\scriptsize  $n_7n_8\neq 0 $}	\\

{\scriptsize  $ \phi=-(n_7n_8)X^1\wedge X^3\wedge X^4$} & {\scriptsize  $n(X_2)=n_1X_2+n_4X_3$}     &	{\scriptsize  $$} \\

	{\scriptsize  $ $} & {\scriptsize  $n(X_3)=n_6X_3+n_7X_4$} &{\scriptsize  $$}  \\

{\scriptsize  $$} &{\scriptsize  $ n(X_4)=n_8X_3+n_1X_4$} &{\scriptsize  $$}  \\ 
\hline
	{\scriptsize  $r=X_1 \wedge X_2$}& {\scriptsize$n(X_1)=n_1X_1+n_2X_3+n_3X_4$}  &{\scriptsize  $n_5n_8\neq 0 $} 	\\

{\scriptsize  $ \phi=(n_5n_8)X^2\wedge X^3\wedge X^4$}& {\scriptsize  $n(X_2)=n_1X_2+n_4X_3+n_5X_4$}     &	{\scriptsize  $$} \\

	{\scriptsize  $ $}& {\scriptsize  $n(X_3)=n_6X_3$}  &{\scriptsize  $$}  \\

{\scriptsize  $$} &{\scriptsize  $ n(X_4)=n_8X_3+n_1X_4$} &{\scriptsize  $$}  \\ 
\hline
	{\scriptsize  $r=X_1 \wedge X_2$}& {\scriptsize$n(X_1)=n_1X_1+n_2X_3+n_3X_4$} & {\scriptsize  $n_4n_6\neq n_4n_9 $}\\

{\scriptsize  $ \phi=(n_1n_4-n_4n_9)X^2\wedge X^3\wedge X^4$}& {\scriptsize  $n(X_2)=n_1X_2+n_4X_3+n_5X_4$}    &	{\scriptsize  $$} \\

	{\scriptsize  $ $} & {\scriptsize  $n(X_3)=n_1X_3+n_7X_4$} &{\scriptsize  $$}  \\

	{\scriptsize  $$} &{\scriptsize  $ n(X_4)=n_9X_4$} &{\scriptsize  $$}  \\ 
\hline
	{\scriptsize  $r=X_1 \wedge X_2$} & {\scriptsize$n(X_1)=n_1X_1+n_2X_3+n_3X_4$}  &{\scriptsize  $n_7n_8\neq 0 $} \\

{\scriptsize  $ \phi=-(n_7n_8)X^1\wedge X^3\wedge X^4$} & {\scriptsize  $n(X_2)=n_1X_2+n_4X_3+n_5X_4$}   &{\scriptsize  $n_8n_5=n_4n_9-n_4n_1 $}\\

	{\scriptsize  $ $} & {\scriptsize  $n(X_3)=n_1X_3+n_7X_4$} &{\scriptsize  $$}  \\

{\scriptsize  $$}&{\scriptsize  $ n(X_4)=n_8X_3+n_9X_4$}  &{\scriptsize  $$}  \\ 
\hline
	{\scriptsize  $r=X_1 \wedge X_2$}& {\scriptsize$n(X_1)=n_1X_1+n_2X_3+n_3X_4$}  &{\scriptsize  $ n_8n_5\neq n_4n_9-n_4n_1$} \\

{\scriptsize  $ \phi=(n_5n_8-n_1n_4-n_4n_9)X^2\wedge X^3\wedge X^4$}& {\scriptsize  $n(X_2)=n_1X_2+n_4X_3+n_5X_4$}     &	{\scriptsize  $$} \\

	{\scriptsize  $$} & {\scriptsize  $n(X_3)=n_1X_3$} &{\scriptsize  $$}  \\

	{\scriptsize  $$} &{\scriptsize  $ n(X_4)=n_8X_3+n_9X_4$} &{\scriptsize  $$}  \\ 
\hline

 	{\scriptsize  $r=X_1 \wedge X_2-X_1 \wedge X_3+X_1 \wedge X_4$}& {\scriptsize$n(X_1)=(n_2-n_1)X_1$} & {\scriptsize  $n_1\neq 0$}\\

{\scriptsize  $ \phi=-(n_1^2+n_1n_6)X^2\wedge X^3\wedge X^4$} & {\scriptsize  $n(X_2)=(c_3-c_4)X_1+(n_6+n_2)X_2$}     & {\scriptsize  $n_1\neq n_6 $}\\

	{\scriptsize  $$} & {\scriptsize  $n(X_3)=c_3X_1+(n_1+n_5+n_6)X_2$} &{\scriptsize  $$}  \\ 

	{\scriptsize  $$} & {\scriptsize  $\quad \quad \quad \quad+(n_2-n_1+n_6)X_3+n_1X_4$} &{\scriptsize  $$}  \\ 

{\scriptsize  $$}&{\scriptsize  $ n(X_4)=c_4X_1+n_5X_2+n_6X_3+n_2X_4$}  &{\scriptsize  $$}  \\ 
\hline
\end{tabular}
\end{center}
\newpage
\begin{center}
{\footnotesize   \bf{Table 3.b.}} \label{TT}
{\small Equivalence classes of $r$-$qn$ structures on four-dimensional non-symplectic real Lie algebra $A_{4,8}$}
\begin{tabular}{ | l  l | l | p{14mm} }
\hline\hline

{\scriptsize  $r$-matrix $r$}
&{\scriptsize $(1,1)$-tensor field $n$ }	&{\scriptsize $Comments$ }\\
{\scriptsize  $ $}
&{\scriptsize $3$-form $\phi$ }	&{\scriptsize $$ }\\
\hline
\smallskip
{\scriptsize  $r=X_2 \wedge X_4$} &{\scriptsize$n(X_1)=n_1X_1$}  &  	{\scriptsize $$}\\

{\scriptsize  $ \phi\equiv 0$}  & {\scriptsize  $n(X_2)=n_2X_2$}    &	{\scriptsize  $ $}\\

{\scriptsize  $$}& {\scriptsize  $n(X_3)=n_1X_3$}  &{\scriptsize  $$}  \\

{\scriptsize  $$} &{\scriptsize  $ n(X_4)=n_2X_4$}&	{\scriptsize  $$}   \\ 
\hline

{\scriptsize  $r=X_2 \wedge X_4$} &{\scriptsize$n(X_1)=n_1X_1$}  &  	{\scriptsize $c_3\in\mathbb R -\{0\}$}\\

{\scriptsize  $ \phi\equiv 0$}  & {\scriptsize  $n(X_2)=n_2X_2$}    &	{\scriptsize  $ $}\\

{\scriptsize  $$}& {\scriptsize  $n(X_3)=c_3X_2+n_1X_3$}  &{\scriptsize  $$}  \\

{\scriptsize  $$} &{\scriptsize  $ n(X_4)=n_2X_4$}&	{\scriptsize  $$}   \\ 
\hline

{\scriptsize  $r=X_2 \wedge X_4$} &{\scriptsize$n(X_1)=n_1X_1$}  &  	{\scriptsize $c_4\in\mathbb R -\{0\}$}\\

{\scriptsize  $ \phi\equiv 0$}  & {\scriptsize  $n(X_2)=n_2X_2$}    &	{\scriptsize  $c_3\in\mathbb R $}\\

{\scriptsize  $$}& {\scriptsize  $n(X_3)=c_3X_2+n_1X_3+c_4X_4$}  &{\scriptsize  $$}  \\

{\scriptsize  $$} &{\scriptsize  $ n(X_4)=n_2X_4$}&	{\scriptsize  $$}   \\ 
\hline

{\scriptsize  $r=X_3 \wedge X_4$} &{\scriptsize$n(X_1)=n_1X_1$}  &  	{\scriptsize $$}\\

{\scriptsize  $ \phi\equiv 0$}  & {\scriptsize  $n(X_2)=n_1X_2$}    &	{\scriptsize  $$}\\

{\scriptsize  $$}& {\scriptsize  $n(X_3)=n_3X_3$}  &{\scriptsize  $$}  \\

{\scriptsize  $$} &{\scriptsize  $ n(X_4)=n_3X_4$}&	{\scriptsize  $$}   \\ 
\hline

{\scriptsize  $r=X_3 \wedge X_4$} &{\scriptsize$n(X_1)=n_1X_1$}  &  	{\scriptsize $c_2\in\mathbb R -\{0\}$}\\

{\scriptsize  $ \phi\equiv 0$}  & {\scriptsize  $n(X_2)=n_1X_2+c_2X_3$}    &	{\scriptsize  $$}\\

{\scriptsize  $$}& {\scriptsize  $n(X_3)=n_3X_3$}  &{\scriptsize  $$}  \\

{\scriptsize  $$} &{\scriptsize  $ n(X_4)=n_3X_4$}&	{\scriptsize  $$}   \\ 
\hline

{\scriptsize  $r=X_3 \wedge X_4$} &{\scriptsize$n(X_1)=n_1X_1$}  &  	{\scriptsize $c_4\in\mathbb R -\{0\}$}\\

{\scriptsize  $ \phi\equiv 0$}  & {\scriptsize  $n(X_2)=n_1X_2+c_2X_3+c_4X_4$}    &	{\scriptsize  $c_2\in\mathbb R $}\\

{\scriptsize  $$}& {\scriptsize  $n(X_3)=n_3X_3$}  &{\scriptsize  $$}  \\

{\scriptsize  $$} &{\scriptsize  $ n(X_4)=n_3X_4$}&	{\scriptsize  $$}   \\ 
\hline

{\scriptsize  $r=X_1 \wedge X_2+X_1 \wedge X_3+X_1 \wedge X_4$} &{\scriptsize$n(X_1)=(n_1+n_2)X_1$}  &  	{\scriptsize $n_2\in\mathbb R -\{0\}$}\\

{\scriptsize  $ \phi=-n_2n_5X^2\wedge X^3\wedge X^4$}  & {\scriptsize  $n(X_2)=-(n_3+n_4)X_1+(n_1+n_2-n_5)X_2$}    &	{\scriptsize  $n_5\in\mathbb R -\{0\} $}\\

{\scriptsize  $$}& {\scriptsize  $n(X_3)=n_3X_1+n_1X_3$}  &{\scriptsize  $$}  \\

{\scriptsize  $$} &{\scriptsize  $ n(X_4)=n_4X_1+n_5X_2+n_2X_3+(n_1+n_2)X_4$}&	{\scriptsize  $$}   \\ 
\hline

{\scriptsize  $r=X_1 \wedge X_2+X_1 \wedge X_3$} &{\scriptsize$n(X_1)=(2n_1-n_2)X_1$}  &  	{\scriptsize $n_1\neq n_2$}\\

{\scriptsize  $ \phi=2(n_1-n_2)^2X^1\wedge X^2\wedge X^4$}  & {\scriptsize  $n(X_2)=-n_3X_1+n_1X_2+(n_1-n_2)X_3$}    &	{\scriptsize  $ n_5\neq -n_6$}\\

{\scriptsize  $\quad\quad-2(n_1-n_2)^2X^1\wedge X^3\wedge X^4 $}& {\scriptsize  $n(X_3)=n_3X_1+(n_1-n_2)X_2+n_1X_3$}  &{\scriptsize  $$}  \\

{\scriptsize  $\quad\quad +(n_2-n_1)(n_5+n_6)X^2\wedge X^3\wedge X^4$} &{\scriptsize  $ n(X_4)=n_4X_1+n_5X_2+n_6X_3+n_2X_4$}&	{\scriptsize  $$}   \\ 
\hline

{\scriptsize  $r=X_1 \wedge X_2$} &{\scriptsize$n(X_1)=n_1X_1$}  &  	{\scriptsize $n_5\in \mathbb R-\{0\}$}\\

{\scriptsize  $ \phi=(n_5n_6-n_1n_5)X^2\wedge X^3\wedge X^4$}  & {\scriptsize  $n(X_2)=n_1X_2$}    &	{\scriptsize  $ n_1\neq n_6$}\\

{\scriptsize  $ $}& {\scriptsize  $n(X_3)=n_2X_1+n_6X_3$}  &{\scriptsize  $$}  \\

{\scriptsize  $$} &{\scriptsize  $ n(X_4)=n_3X_1+n_5X_2+n_7X_3+n_1X_4$}&	{\scriptsize  $$}   \\ 
\hline

{\scriptsize  $r=X_1 \wedge X_2$} &{\scriptsize$n(X_1)=n_1X_1$}  &  	{\scriptsize $n_4\in \mathbb R-\{0\}$}\\

{\scriptsize  $ \phi=(-n_4n_7)X^2\wedge X^3\wedge X^4$}  & {\scriptsize  $n(X_2)=n_1X_2$}    &	{\scriptsize  $ n_7\in \mathbb R-\{0\}$}\\

{\scriptsize  $ $}& {\scriptsize  $n(X_3)=n_2X_1+n_4X_2+n_1X_3$}  &{\scriptsize  $$}  \\

{\scriptsize  $$} &{\scriptsize  $ n(X_4)=n_3X_1+n_5X_2+n_7X_3+n_1X_4$}&	{\scriptsize  $$}   \\ 
\hline


{\scriptsize  $r=X_1 \wedge X_3$} &{\scriptsize$n(X_1)=n_1X_1$}  &  	{\scriptsize $n_5\in \mathbb R-\{0\}$}\\

{\scriptsize  $ \phi=(-n_5n_6)X^2\wedge X^3\wedge X^4$}  & {\scriptsize  $n(X_2)=n_2X_1+n_1X_2+n_6X_3$}    &	{\scriptsize  $ n_6\in \mathbb R-\{0\}$}\\

{\scriptsize  $ $}& {\scriptsize  $n(X_3)=n_1X_3$}  &{\scriptsize  $$}  \\

{\scriptsize  $$} &{\scriptsize  $ n(X_4)=n_3X_1+n_5X_2+n_7X_3+n_1X_4$}&	{\scriptsize  $$}   \\ 
\hline

{\scriptsize  $r=X_1 \wedge X_3$} &{\scriptsize$n(X_1)=n_1X_1$}  &  	{\scriptsize $n_7\in \mathbb R-\{0\}$}\\

{\scriptsize  $ \phi=(n_4n_7-n_1n_7)X^2\wedge X^3\wedge X^4$}  & {\scriptsize  $n(X_2)=n_2X_1+n_4X_2$}    &	{\scriptsize  $n_1\neq n_4 $}\\

{\scriptsize  $ $}& {\scriptsize  $n(X_3)=n_1X_3$}  &{\scriptsize  $$}  \\

{\scriptsize  $$} &{\scriptsize  $ n(X_4)=n_3X_1+n_5X_2+n_7X_3+n_1X_4$}&	{\scriptsize  $$}   \\ 
\hline

\end{tabular}
\end{center}

\section{Some remarks on $r$-$qn$ structures}\label{app}
In this section we shall consider the conditions for which an $r$-$qn$ structure on Lie algebra $\mathfrak g$ defines a generalized complex structure on $\mathfrak g$ or an $R$-matrix on double of Lie algebra  $\mathfrak g \oplus\mathfrak g^*$. We bring some relevant examples of $r$-$qn$ structures of previous section. 

It is well-known that for a Lie bialgebra $(\mathfrak g, \mathfrak g^*)$, the vector space $\mathfrak g \oplus\mathfrak g^*$, so-called the {\em double of Lie bialgebra}, is equipped with a Lie algebra structure defined by
\begin{equation}\label{big}
[X+\alpha,Y+\beta]=[X,Y]+[\alpha, \beta]+ad^*_{X}\beta+ad^*_{\alpha}Y-ad^*_{Y}\alpha-ad^*_{\beta}X,\quad \forall X,Y\in \mathfrak g,\quad \alpha,\beta\in \mathfrak g^*.
\end{equation}

where $ad^*_{X}\alpha$ and $ad^*_{\alpha}X$ are the coadjoint representations of $\mathfrak g$ on $\mathfrak g^*$  and of $\mathfrak g^*$ on $\mathfrak g$ respectively. Let us use $\mathfrak g \bowtie \mathfrak g^*$ to denote the vector space $\mathfrak g \oplus\mathfrak g^*$ with the Lie algebra structure (\ref{big}). Recall that the nondegenerate symmetric bilinear form on the vector space $\mathfrak g \oplus\mathfrak g^*$ defined by
\[
\left\langle X+\alpha,Y+\beta\right\rangle =\alpha(X)+\beta(Y),\quad \forall X,Y\in \mathfrak g, \;\forall \alpha, \beta \in \mathfrak g^*,
\]
for more details we refer to \cite{Ko} and \cite{Lu}. 

By definitions, the solutions of modified Yang-Baxter equation on double Lie bialgebra $\mathfrak g \bowtie \mathfrak g^*$ with coefficient $k=-1$ are identified with the generalized complex structures on the Lie algebra $\mathfrak g$.

\begin{corollary}\label{result1}
The $r$-$qn$ structure $(r, \phi, n)$ on $\mathfrak g$ defines an $R$-matrix $\mathit J$ on the Lie algebra $\mathfrak g \bowtie \mathfrak g^*$ of the form 
\begin{equation}\label{J1}
\mathit{J}=\left(\begin{array}{cc}
n& r^{\sharp}\\ \theta_{\sharp} & -n^t\end{array} \right)
\end{equation}

\noindent if $\phi$ is a $3$-cobondary and following two conditions hold
\begin{equation}\label{con.g.c2}
\begin{array}{rcl}
n^t\theta_{\sharp}&=&\theta_{\sharp}N\\[3pt]

n^2+r^{\sharp}\theta_{\sharp}&=&kId,\\[3pt]
\end{array}\end{equation}
where $\theta$ is a $2$-cochain in $C^2(\mathfrak g)$ such that $\phi=\partial\theta$.
\end{corollary}
\begin{proof}
It is proved directly from the Corollary \ref{cr.1}.  

\end{proof}
\subsection*{Example 1.} 
Consider the four-dimensional non-symplectic real Lie algebra $\mathfrak g:=A_{4,8}$ with non-zero commutators $[X_2,X_3]=X_1$, $[X_2,X_4]=X_2$ and $[X_3,X_4]=-X_3$. We choose the third $r$-$qn$ structure $(r,n)$ on this algebra of the table 3, as 
\[
r^{\sharp}(X^1)=X_2+X_3+X_4,\quad r^{\sharp}(X^2)=-X_1,\quad r^{\sharp}(X^3)=-X_1,\quad r^{\sharp}(X^4)=-X_1,
\]
\[
\phi=-(n_2n_5)X^2\wedge X^3\wedge X^4,
\]
\[
\begin{array}{rclcrclcrcl}
n(X_1)&=& (n_1+n_2)X_1,&&  n_(X_2)&=&-(n_3+n_4)X_1+(n_1+n_2-n_5)X_2,\\[3pt]
n(X_3)&=& n_3X_1+n_1X_3, && n(X_4)&=& n_4X_1+n_5X_2+n_2X_3+(n_1+n_2)X_4.\\
\end{array}
\]

The Lie algebra structure on the dual Lie algebra $\mathfrak g^*$ of $\mathfrak g$ induced by $r$-matrix $r$ has the non-zero commutators $[X^1,X^2]=X^2-X^4$ and $[X^1,X^3]=X^4-X^3$.

Integrating $3$-cocycle $\phi$ we find out it is a $3$-coboundary, that is there exist $2$-cochain $\theta$ such that $\partial \theta=\phi$, where
\[
\theta=(\frac{1}{2}n_2n_5)X^1\wedge X^4-(\frac{1}{2}n_2n_5)X^2\wedge X^3.
\]

Applying $n$, $r^{\sharp}$ and $\theta_{\sharp}$ on two conditions (\ref{con.g.c2}) and solving the equations, we get
\[
n_2=n_3=n_5=-4n_1,\quad n_4=0,\quad k=n_1^2.
\]
Thus the linear map $\mathit J:\mathfrak g \bowtie \mathfrak g^*\to \mathfrak g \bowtie \mathfrak g^*$ is a solution of MYBE with coefficient $k=n_1^2$, on the Lie algebra $\mathfrak g \bowtie \mathfrak g^*$ as

\begin{equation}
\mathit J=\left( \!\!\!\!‎‎
\begin{tabular}{c}‎
\begin{tabular}{cccc|cccc}
$-3n_1$& $4n_1$& $-4n_1$& 0& 0& $-1$& $-1$& 1\\‎
0& $n_1$& 0& $-4n_1$& 1& 0& 0& 0\\‎
0& 0&$ n_1$& $-4n_1$& 1& 0& 0& 0\\‎
0& 0& 0 &$ -3n_1$& 1& 0& 0& 0‎\\
\hline
0& 0& 0& $-8n_1^2$ & $-3n_1$& & 0& 0\\‎
0& 0& $8n_1^2$& 0&  $4n_1$& $n_1$& 0& 0\\‎
0& $-8n_1^2$& 0& 0& $-4n_1$& 0& $n_1$& 0\\‎
$8n_1^2$& 0& 0& 0& 0& $- 4n_1$& $- 4n_1$& $-3n_1$
\end{tabular}‎‎
\end{tabular}‎\!\!\!\! \right).
\end{equation}

\begin{remark}\label{g.c2}
Under the same assumption as in the Corollary \ref{result1}, if $k=-1$, then $\mathit J$ defines a generalized complex structure on $\mathfrak g$.

\end{remark}
\subsection*{Example 2.}
Consider the four-dimensional symplectic real Lie algebra $\mathfrak g:=A_{4,1}$ with non-zero commutators $[X_2,X_4]=X_1$ and $[X_3,X_4]=X_2$. On this algebra there is an $r$-$n$ structure (trivial $r$-$qn$ structure) $(r,n)$ (see table 1), as
\[
r^{\sharp}(X^1)=X_4,\quad r^{\sharp}(X^2)=-X_3,\quad r^{\sharp}(X^3)=X_2,\quad r^{\sharp}(X^4)=-X_1,
\]
Note that, the Lie algebra structure on the dual Lie algebra $\mathfrak g^*$ of $\mathfrak g$ induced by $r$-matrix $r$ has the non-zero commutators $[X^1,X^2]=X^3$ and $[X^1,X^3]=X^4$.

\noindent The Nijenhuis operator compatible with $r$ is characterized by
\[
\begin{array}{rclcrclcrcl}
n(X_1)&=& n_1X_1,&&  n_(X_2)&=&-n_2X_1+n_3X_2,\\[3pt]
n(X_3)&=& n_4X_1+n_3X_3, && n(X_4)&=& n_4X_2+n_2X_3+n_1X_4.\\
\end{array}
\]
With straightforward computation we see there is no $3$-coboundary on this Lie algebra. One can check all $2$-cochains $\theta\in C^2(\mathfrak g)$ which satisfy cocycle condition are the form  
\[
\theta=cX^1\wedge X^4+dX^2\wedge X^3+eX^2\wedge X^4+fX^3\wedge X^4, \quad a,b,d,e,f\in \mathbb R,
\]
that is $\phi:=\partial \theta=0$. So we have
\[
\begin{array}{rclcrclcrcl}
\theta_{\sharp}(X_1)&=& cX^4,&& \theta_{\sharp}(X_2)&=& dX^3+eX^4,\\[3pt] \theta_{\sharp}(X_3)&=&-dX^2+fX^4, &&\theta_{\sharp}(X_4)&=& -cX^1-eX^2-fX^3.\\
\end{array}
\] 
Applying $n$, $r^{\sharp}$ and $\theta_{\sharp}$ on two conditions (\ref{con.g.c2}), for $k:=-1$ we get
\[
c=n_1^2+1,\quad d=-1-n_1^2,\quad e=-2n_1n_2,\quad f=2n_1n_4  ,\quad n_1=n_3.
\]
Therefore, we have a generalized complex structure $\mathit J:\mathfrak g \oplus\mathfrak g^*\to \mathfrak g \oplus\mathfrak g^*$ on the Lie algebra $\mathfrak g:=A_{4,1}$ as \\
\begin{equation}
\mathit J=\left( \!\!\!\!‎‎
\begin{tabular}{c}‎
\begin{tabular}{cccc|cccc}
$n_1$& $-n_2$& $n_4$& 0& 0& 0& 0& $-1$\\‎
0& $n_1$& 0& $n_4$& 0& 0& 1&0\\‎
0& 0& $n_1$&$ n_2$& 0& $-1$& 0& 0\\‎
0& 0& 0 & $n_1$& 1& 0& 0& 0‎\\
\hline
0& 0& 0& $-(n_1^2+1)$& $-n_1$& 0& 0& 0\\‎
0& 0& $n_1^2+1$& $2n_1n_2$& $n_2$& $-n_1$& 0& 0\\‎
0& $-(n_1^2+1)$& 0& $-2n_1n_4$& $-n_4$& 0& $-n_1$& 0\\‎
$n_1^2+1$& $-2n_1n_2$& $2n_1n_4$& 0& 0& $-n_4$& $-n_2$& $-n_1$
\end{tabular}‎‎
\end{tabular}‎\!\!\!\! \right).
\end{equation}

\begin{remark}\label{g.c3}
The trivial $r$-$qn$ structures ($r$-$n$ structures) on $\mathfrak g$ for which $\theta\equiv 0$, the structure $\mathit J$ defined in (\ref{J1}) is a generalized complex structure on $\mathfrak g$ if $n^2=-Id$. 
\end{remark}

\subsection*{Example 3.} Consider the four-dimensional symplectic real Lie algebra $\mathfrak g:=II\oplus\mathbb R$ with non-zero commutator $[X_2,X_3]=X_1$. We have an $r$-$n$ structure $(r,n)$ (trivial $r$-$qn$ structure) on this algebra with he following $r$-matrix $r$ and Nijenhuis operator $n$ (listed in the table 2 of \cite{Zohreh}).

\noindent The non-degenerate $r$-matrix $r$ is as $r=X_1\wedge X_3-X_2\wedge X_4$, so
\[
r^{\sharp}(X^1)=X_3,\quad r^{\sharp}(X^2)=-X_4,\quad r^{\sharp}(X^3)=-X_1,\quad r^{\sharp}(X^4)=X_2,
\]
Note that, the Lie algebra structure on the dual Lie algebra $\mathfrak g^*$ of $\mathfrak g$ induced by $r$-matrix $r$ has the non-zero commutator $[X^1,X^4]=X^3$.

\noindent The Nijenhuis operator compatible with $r$ is characterized by
\[
\begin{array}{rclcrclcrcl}
n(X_1)&=& n_1X_1+n_5X_4,&&  n_(X_2)&=&-n_2X_1+n_4X_2+n_5X_3,\\[3pt]
n(X_3)&=& n_3X_2+n_1X_3+n_2X_4, && n(X_4)&=& n_3X_1+n_4X_4.\\
\end{array}
\]
One can check in general, there is no $3$-coboundary on this Lie algebra. Since $\phi\equiv 0$, we see that following $2$-cochain $\theta \in C^2(\mathfrak g)$ 
\[
\theta=aX^1\wedge X^2+bX^1\wedge X^3+dX^2\wedge X^3+eX^2\wedge X^4+fX^3\wedge X^4, \quad a,b,d,e,f\in \mathbb R,
\]
are the possibilities for $\phi:=\partial \theta=0$. So we have
\[
\begin{array}{rclcrclcrcl}
\theta_{\sharp}(X_1)&=& aX^2+bX^3,&& \theta_{\sharp}(X_2)&=&-aX^1+dX^3+eX^4,\\[3pt] \theta_{\sharp}(X_3)&=&-bX^1-dX^2+fX^4, &&\theta_{\sharp}(X_4)&=& -eX^2-fX^3.\\
\end{array}
\] 
Applying $n$, $r^{\sharp}$ and $\theta_{\sharp}$ on two conditions (\ref{con.g.c2}) for $k=-1$, we get
\[
a=b=d=e=f=0,\quad n_1=n_4=0, \quad n_3=1,\quad n_5=-1.
\]
 Therefore, we have a generalized complex structure $\mathit J:\mathfrak g \oplus\mathfrak g^*\to \mathfrak g \oplus\mathfrak g^*$ on the Lie algebra $\mathfrak g:=II\oplus \mathbb R$ as 
\begin{equation}
\mathit J=\left( \!\!\!\!‎‎
\begin{tabular}{c}‎
\begin{tabular}{cccc|cccc}
0&$-n_2$& 0& 1& 0& 0&$ -1$& 0\\‎
0& 0& 1& 0& 0& 0& 0& 1\\‎
0&$-1$& 0& 0& 1& 0& 0& 0\\‎
$-1$& 0& $n_2$ & 0& 0& $-1$& 0& 0‎\\
\hline
0& 0& 0& 0& 0& 0& 0 &1\\‎
0& 0& 0& 0& $n_2$& 0& 1& 0\\‎
0& 0& 0& 0& 0& $-1$& 0& $-n_2$\\‎
0& 0& 0& 0&$ -1$& 0& 0& 0
\end{tabular}‎‎
\end{tabular}‎\!\!\!\! \right)
\end{equation}
Note that, the statement of Remark \ref{g.c3} holds since $n^2=-Id$ and $\theta\equiv 0$.



\subsection*{Acknowledgments}
The second author would like to thank Institute of Mathematics Polish Academy of Sciences for the hospitality in a visit where a part of this project was being done.

 \end{document}